\documentclass[10pt,twocolumn]{IEEEtran}
\usepackage{verbatim}
\usepackage{amssymb}  
\usepackage{color}
\usepackage{bm}
\usepackage{xcolor}
\usepackage{algorithm}
\usepackage{algorithmic}
\usepackage{theorem}  

\newcommand{\tabincell}[2]{\begin{tabular}{@{}#1@{}}#2\end{tabular}}

\newcommand{\secref}[1]{Sec. \ref{#1}}

\usepackage{cite}

\ifCLASSINFOpdf
  \usepackage[pdftex]{graphicx}
\else
  \usepackage[dvips]{graphicx}
\fi
\usepackage{amsmath}

\usepackage{algorithmic}

\usepackage{array}

\ifCLASSOPTIONcompsoc
  \usepackage[caption=false,font=normalsize,labelfont=sf,textfont=sf]{subfig}
\else
  \usepackage[caption=false,font=footnotesize]{subfig}
\fi

\usepackage{stfloats}

\usepackage{url}

\theoremstyle{plain}
\newtheorem{remark}{Remark}
\theoremstyle{plain}
\newtheorem{theorem}{Theorem}
\theoremstyle{plain}
\newtheorem{proposition}{Proposition}
\theoremstyle{plain}

\theoremstyle{plain}
\newtheorem{assumption}{Assumption}
\theoremstyle{plain}
\newtheorem{lemma}{Lemma}

\begin{document}

\def\QEDclosed{\mbox{\rule[0pt]{1.3ex}{1.3ex}}}
\def\QEDopen{{\setlength{\fboxsep}{0pt}\setlength{\fboxrule}{0.2pt}\fbox{\rule[0pt]{0pt}{1.3ex}\rule[0pt]{1.3ex}{0pt}}}}
\def\QED{\QEDopen}
\def\proof{}
\def\endproof{\hspace*{\fill}~\QED\par\endtrivlist\unskip}

\title{A Partial Reciprocity-based Channel \\
Prediction Framework for FDD Massive \\
MIMO with High Mobility}

\author{Ziao~Qin, Haifan~Yin, Yandi~Cao, Weidong~Li, and David~Gesbert,~\IEEEmembership{Fellow,~IEEE} 
\thanks{Z. Qin, H. Yin, Y. Cao and W. Li are with Huazhong University of Science and Technology,
430074 Wuhan, China (e-mail: ziao\_qin@hust.edu.cn, yin@hust.edu.cn, yandicao@hust.edu.cn, weidongli@hust.edu.cn).}
\thanks{D. Gesbert is with EURECOM, 06410 Biot, France (e-mail: gesbert@
eurecom.fr).}
\thanks{This work was supported by the National Natural Science Foundation of China under Grant 62071191. The corresponding author is Haifan Yin.}}

\maketitle

\begin{abstract}

Massive multiple-input multiple-output (MIMO) is believed to deliver unrepresented spectral efficiency gains for 5G and beyond. However, a practical challenge arises during its commercial deployment, which is known as the ``curse of mobility''. The performance of massive MIMO drops alarmingly when the velocity level of user increases. In this paper, we tackle the problem in frequency division duplex (FDD) massive MIMO with a novel Channel State Information (CSI) acquisition framework. A joint angle-delay-Doppler (JADD) wideband precoder is proposed for channel training. Our idea consists in the exploitation of the partial channel reciprocity of FDD and the angle-delay-Doppler channel structure. More precisely, the base station (BS) estimates the angle-delay-Doppler information of the UL channel based on UL pilots using Matrix Pencil (MP) method. It then computes the wideband JADD precoders according to the extracted parameters. Afterwards, the user estimates and feeds back some scalar coefficients for the BS to reconstruct the predicted DL channel. Asymptotic analysis shows that the CSI prediction error converges to zero when the number of BS antennas and the bandwidth increases. Numerical results with industrial channel model demonstrate that our framework can well adapt to high speed (350 km/h), large CSI delay (10 ms) and channel sample noise.

\end{abstract}

\begin{IEEEkeywords}
Massive MIMO, curse of mobility, channel prediction, FDD, angle-delay domain, partial reciprocity, Matrix Pencil, 5G.
\end{IEEEkeywords}

\IEEEpeerreviewmaketitle

\section{Introduction}\label{sec1}

\IEEEPARstart{T}{he} 5G wireless communication is being deployed in real life and is given great expectations on high throughput rate, low latency and high reliability. To achieve such intriguing merits of 5G, massive MIMO technology is indispensable. These benefits are brought by the large numbers of antennas at the BS side while eliminating uncorrelated noise and fast fading \cite{2010MazattaMIMO}. Massive MIMO system has shown great potential in improving spectral efficiency (SE) and energy efficiency (EE) \cite{2013MazattaSE}. Even though the pilot contamination problem limits massive MIMO system performance \cite{2011MazettaPilot}, this effect can be mitigated by exploiting the angular structure of channel \cite{2013YinJSAC} and differences of the channel power \cite{2014BlindRuf}.

High SE performance depends on accurate CSI. Thanks to the channel reciprocity of TDD, CSI can be obtained by an acceptable pilot training overhead which scales with the number of user equipments (UEs) instead of the number of BS antennas. Therefore, TDD mode may be the favorable choice for massive MIMO system. However, a large percent of current cellular communication system operates in FDD mode, thus massive MIMO operating in FDD mode has equal importance. The authors in \cite{2018TDDvsFDD} measured the performance at 2.6 GHz in the two modes and conclude that each enjoys its own advantages in different scenarios. 
Unfortunately, the CSI acquisition in FDD mode is more challenging due to the non-reciprocal UL and DL channel, and therefore the training and feedback overhead. Many research works have offered possible solutions to CSI acquisition in FDD massive MIMO system. The authors in \cite{2013JSDM} and \cite{2020JSDM} utilized a statistical channel information and user grouping based prebeamformer to reduce pilot training and feedback overhead, which is known as the ``Joint  spatial  division  and  multi-plexing'' (JSDM) method. The low-rankness property of channel correlation matrices was considered to design pilot training and feedback under a spatial correlation channel model in \cite{2015spatial}. Spatial sparsity of massive MIMO channel can also be exploited through the compressed sensing (CS) method. In \cite{2019mmimoCS}, the authors estimated the channel by extracting channel parameters through CS method in millimeter-wave massive MIMO. Another possible approach to address FDD massive MIMO channel estimation is based on channel parameter extraction. By exploiting angle information through a discrete Fourier transform (DFT) projection \cite{2017gaofeiJSAC}, the DL channel can be reconstructed through angular information and channel gain which are estimated separately. The authors in \cite{2019hanEfficient} introduced a Newtonzied orthogonal matching pursuit (NOMP) method to detect angle, delay and gains and reconstructed the channel following a multipath channel model. Deep learning method was also utilized to reconstruct the DL channel \cite{2020DeepHan}. However, the papers above mainly considered a block fading scenario where the channel was {assumed to be} constant for a period of time. This assumption is reasonable in a stationary or low-mobility scenario. 

However, in practice, the system performance may degrade badly in mobility scenarios \cite{2019mobilityReport, 2020yinMobility} even in TDD mode. This effect is caused by the time-varying nature of channel. The outdated CSI severely corrupts the SE performance. Unlike in stationary settings, Doppler frequency shift becomes nontrivial in {mobile} environments. The authors in  \cite{2020Dataprediction} proposed a data-aided channel prediction based on variational Bayesian inference (VBI) framework in high mobility scenario. A maximum-likelihood based method is introduced in \cite{2017V2V} to estimate channel parameters in vehicle-to-vehicle (V2V) MIMO system. Some works addressed the CSI delay influence on the channel in a theoretical view \cite{2012csidelay,2014csidelay}. In \cite{2020yinMobility}, the authors proposed a channel prediction method to solve the the mobility problem utilizing Prony-based angle-delay domain channel prediction. The authors of \cite{2021dl_mobility} addressed the mobility problem in massive MIMO from a deep learning view. Nevertheless these papers mainly focused on TDD mode. 

Different from TDD mode, 
{closed-loop feedback of CSI from the UE to the BS is inevitable, which introduces CSI quantization error, in addition to even larger CSI delay. Especially in high-mobility scenario, the channel coherence time is much shorter than the low-mobility scenario and timely feedback is more challenging in FDD mode due to the different operating frequency bands between UL and DL.} Worse still, the training and feedback overhead are much heavier than TDD, and thus has to be reduced. The state-of-the-art algorithms like CS \cite{2019mmimoCS}, deep learning \cite{2021dl_mobility} and JSDM \cite{2020JSDM} method mainly focus on reducing the pilot training and feedback overhead. Some research works utilized maximum-likelihood method \cite{2017V2V}, deep learning \cite{2021dl_mobility} and machine learning method \cite{2021KF_ML} to address the CSI aging problem in TDD. 
{The authors in \cite{2021TVT} utilized partial channel reciprocity in terms of the angular support to facilitate the CSI feedback in TDD for the case that the UE has unequal number of TX and RX antennas. 
In \cite{2021reconstruct}, a channel reconstruction method based on CSI-RS and  SRS in TDD system was proposed.} 
However, these methods did not consider the mobility problem in FDD and the given solutions were mostly NP-hard. As the high mobility demands timely CSI acquisition and high efficiency of channel prediction algorithm, these methods cannot directly apply in FDD massive MIMO with high-mobility. {Recently, some works like \cite{2019jiangtaoJSTSP,2020DeepHan} proposed channel prediction methods for FDD massive MIMO. Unfortunately, the performances of these methods may not be guaranteed in a rich scattering environment with a large number of multipath, especially in high-mobility scenario.} To the best of our knowledge, few works have addressed these real-world problems simultaneously in a practical multipath channel model. 

In this paper, we aim to solve this problem with a novel CSI acquisition framework which is easy to deploy and has  polynomial complexity. Even though the full channel reciprocity in FDD is not available like TDD, some frequency-unrelated channel parameters are reciprocal between the DL channel and the UL channel  \cite{reciprocity2002spatial,3gpp:36.897,2017GaoAngle}. {Through the channel measurement campaigns, the partial reciprocity in FDD was verified in \cite{2020GlobalCom}}. The partial reciprocity allows us to extract some useful information from the UL channel estimation, e.g., the angle, delay, and Doppler frequency shift. We propose to extract such information through an efficient linear prediction method known as MP \cite{1990MP}. 
Once the information is obtained, we design a JADD spatial-frequency precoders for the wideband DL pilot transmission. The precoders capitalize on the channel sparsity in angle-delay domain, as well as the partial reciprocity. They will help reduce the training overhead and facilitate the DL channel reconstruction. 
Note that different from existing methods, our precoder are wideband and require joint operation from the BS and the UE. Afterwards, the UE estimates some complex scalar coefficients based on the precoded DL training signal and feeds them back to the BS. Finally the BS reconstructs the DL CSI using the coefficients and the extracted UL channel parameters. 

{Different from previous channel reconstruction methods like \cite{2019jiangtaoJSTSP,2019hanEfficient,2021reconstruct,2020DeepHan,2020yinMobility}, we devise a wideband precoder and JADD feedback framework.} Our framework outperforms traditional methods which are typically based on the NP-hard solutions or failing to timely update CSI. Moreover, our approach is capable of predicting the channel in polynomial complexity. Simulation results under the 3rd Generation Partner Project (3GPP) channel model indicate that our proposed framework is robust to high mobility scenarios with even 350 km/h of UE speed and to large CSI delay. Moreover, we test our framework in different scattering environments, BS antenna configurations and noisy channel sample cases. The numerical results demonstrate the robustness of our framework.

Our main contributions are
\begin{itemize}
    \item We address the mobility problem of FDD massive MIMO under an industrial multipath channel model, which was rarely considered in the literature. By exploiting the angle-delay-Doppler structure and the partial reciprocity of the channel, we propose a JADD CSI acquisition framework, which combats the outdated DL CSI and reduces the training overhead simultaneously. 
    
    \item We propose to extract the Doppler frequency shifts using the MP method in angle-delay domain, where the channel shows more sparsity. This method requires less channel samples and achieves high accuracy for the Doppler estimation, due to the high spatial and frequency resolution of a wideband massive MIMO system. 
    
    \item We propose a novel training and feedback framework for FDD massive MIMO. The key ingredients are a wideband precoder for DL pilots and the computation of the complex coefficients of the DL paths at the UE side. This precoding method requires a two-step joint operation of the BS and the UE. Only scalar coefficients need to be fed back to the BS. In this framework, the training and feedback overhead no longer depends on the number of the BS antennas and bandwidth, but on the angle-delay sparsity of the channel and the prediction order of the MP method.
    
    \item We derive the upper bound of the DL channel prediction error under limited BS antennas and bandwidth. Our asymptotic analysis shows the channel prediction error converges to zero when the number of antennas at the BS and the bandwidth increase while only two UL channel samples are needed. We also suggest the choice of the prediction order when applying our method. 
    
\end{itemize}

The rest of the paper is organized as follows. \secref{sec2} introduces our channel model. \secref{sec3} demonstrates the UL channel parameters estimation method. \secref{sec4} discusses  the DL pilot training, feedback, and DL channel reconstruction. \secref{sec5} contains the performance analysis of our proposed framework. \secref{sec6} shows the numerical results of our framework. \secref{sec7} is the conclusion of our work.

Notations: The boldface front stands for vector and matrix. $ \otimes $ is Kronecker product symbol. $\rm{diag}\left(\bf{X}\right)$ means a diagonal matrix with $\bf{X}$ as its diagonal elements and if $\bf{X}$ is a block matrix, $\rm{diag}\left(\bf{X}\right)$ denotes a block diagonal matrix. $\bf{vec\left(\bf{X}\right)}$ is the vectorization of $\bf{X}$. ${{\bf{X}}^\dag },{{\bf{X}}^T},$ and ${{\bf{X}}^H}$ denote the Moore-Penrose inversion, transpose and conjugation of $\bf{X}$, respectively. ${\mathbb{C}^{a \times b}}$ is a matrix space with $a$ rows and $b$ columns. $\left| {\bf{x}} \right|$ denotes the absolute value of $\bf{x}$ and ${\left\| {\bf{X}} \right\|_2}$ is the second-order induced norm of $\bf{X}$. $\Re\left(\alpha\right)$ denotes the real part of complex $\alpha$. ${\rm{mod}}\left(x\right)$ is the modular operation of $x$. $ \buildrel \Delta \over = $ refers to the definition symbol. $\mathbb{E}\left\{ x \right\}$ means calculating expectation of $x$. ${{\bf{R}}} \sim {\mathcal {CN}}\left( {0,{\sigma}^2{\bf{I}}} \right)$ means that ${\bf{R}}$ satisfies zero-mean complex circular Gaussian distribution.
\section{System model} \label{sec2}
This paper considers a wideband FDD massive MIMO system where the BS is equipped with a uniform planar array (UPA). The classical orthogonal frequency division multiplex (OFDM) modulation is adopted with ${N_f}$ sub-carriers and a ${f_\Delta }$ subcarrier spacing. The number of BS antennas is ${N_t}= {N_v}{N_h}$, where ${N_v}$ and ${N_h}$ denote the number of antennas in a row and in a column, respectively. The center frequencies of UL and DL are ${f^u}$ and ${f^d}$, respectively. 

A multi-path channel model following \cite{3gpp901} is adopted in our work. The number of paths of the channel is denoted by $P$. The corresponding parameters of each path $p$ are the complex amplitude ${\beta _p}$, steering vector ${{\boldsymbol{\alpha}}^u}\left( {\theta _{p}^u,\phi _{p}^u} \right)$, Doppler frequency shift ${\omega _p}$, and delay ${\tau _p}$. Therefore, the UL channel between the BS and the UE $k$ at a certain time $t$ and frequency $f$ is 
\begin{equation}\label{ul mutiplath channel }
     {\bf{h}}_{k,r}^u\left( {t,f} \right) = \sum\limits_{p = 1}^P {\beta _{p}^u{{\boldsymbol{\alpha}}^u}\left( {\theta _{p}^u,\phi _{p}^u} \right){e^{ - j2\pi f\tau _{p}^u}}{e^{jw_{p}^ut}}}, 
\end{equation}
where the subscript $k,r$ means the $r$-th antenna of the UE $k$ and the superscript $u$ denotes the UL channel. For simplicity, we drop the subscripts $r$ and $k$ here and afterwards. The UL Doppler frequency shift is defined as $w_p^u = {{v\cos \varphi _p^u{f^u}} \mathord{\left/{\vphantom {{v\cos \varphi _p^u{f^u}} c}} \right.\kern-\nulldelimiterspace} c}$, where $v$ is the velocity of the UE and $\varphi _{p}^u$ is the angle between path $p$ and the 3D velocity vector of the UE $\bf{v}$. $c$ is the speed of light. Denote the zenith angle and azimuth angle by $\theta _{p}^u,\phi _{p}^u$, respectively. Fig. \ref{fig_system model} demonstrates the UPA antenna configuration in 3D-Cartesian coordinate system, the zenith angle $\theta$, the azimuth angle $\phi$, speed direction angle $\varphi$ which is the angle between the path and the velocity vector of the UE. The transmit steering vector is ${{\boldsymbol{\alpha}}^u}\left({\theta _{p}^u,\phi _{p}^u} \right)\in{\mathbb{C}^{{N_t} \times 1}}$ and is modeled as the Kronecker product of the vertical steering vector ${\boldsymbol{\alpha}}_v^u\left({\theta _{p}^u} \right)$ and the horizontal steering vector ${\boldsymbol{\alpha}}_h^u\left( {\theta _{p}^u,\phi _{p}^u} \right)$  
\begin{figure}[!t]
\centering
\includegraphics[width=3.2in]{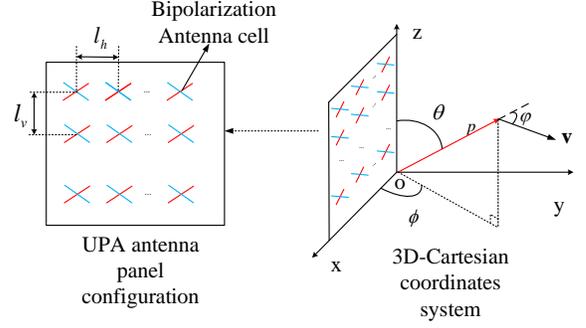}
\caption{UPA antenna configuration in 3D-Cartesian coordinate system.}
\vspace{-0.3cm}
\label{fig_system model}
\end{figure}
\begin{equation}\label{ul steer vector}
    {{\boldsymbol{\alpha}}^u}\left( {\theta _{p}^u,\phi _{p}^u} \right) = {\boldsymbol{\alpha}}_h^u\left( {\theta _{p}^u,\phi _{p}^u} \right) \otimes {\boldsymbol{\alpha}}_v^u\left( {\theta _{p}^u} \right),
\end{equation}
where
\begin{equation}\label{ul horizontal steer vector}
{\boldsymbol{\alpha}}_h^u\left( {\theta _{p}^u,\phi _{p}^u} \right) = \left[ {\begin{array}{*{20}{c}}
1\\
{{e^{j2\pi \frac{{{l_h}{f^u}}}{c}\cos \theta _{p}^u\cos \phi _{p}^u}}}\\
 \cdots \\
{{e^{j2\pi \frac{{{l_h}{f^u}}}{c}\left( {{N_h} - 1} \right)\cos \theta _{p}^u\cos \phi _{p}^u}}}
\end{array}} \right],
\end{equation}
\begin{equation}\label{ul vertical steer vector}
    {\boldsymbol{\alpha}}_v^u\left( {\theta _{p}^u} \right) = \left[ {\begin{array}{*{20}{c}}
1\\
{{e^{j2\pi \frac{{{l_v}{f^u}}}{c}\cos \theta _{p}^u}}}\\
 \cdots \\
{{e^{j2\pi \frac{{{l_v}{f^u}}}{c}\left( {{N_v} - 1} \right)\cos \theta _{p}^u}}}
\end{array}} \right],
\end{equation}
and ${l_v},{l_h}$ are the spacing between the antennas in vertical direction and horizontal direction, respectively. Similarly, the DL channel is modeled as
\begin{equation}\label{dl multipath channel}
     {\bf{h}}^d\left( {t,f} \right) = \sum\limits_{p = 1}^P {\beta _{p}^d{{\boldsymbol{\alpha}}^d}\left( {\theta _{p}^d,\phi _{p}^d} \right){e^{ - j2\pi f\tau _{p}^d}}{e^{ - j2\pi \left( {{f^d} - {f^u}} \right)\tau _{p}^d}}{e^{jw_{p}^dt}}},  
\end{equation}
where $d$ stands for the DL channel. 

Unlike TDD, in FDD only some parameters of the UL and DL channels are reciprocal \cite{3gpp:36.897}
\begin{equation}\label{reciprocity}
    \tau _{p}^u = \tau _{p}^d,\theta _{p}^u = \theta _{p}^d,\phi _{p}^u = \phi _{p}^d,\frac{{w_{p}^u}}{{w_{p}^d}} = \frac{{{f^u}}}{{{f^d}}}.
\end{equation}
The DL steering vector ${{\boldsymbol{\alpha}}^d}\left( {\theta _p^d,\phi _p^d} \right)$ is frequency-related and is calculated by the UL steering vector ${{\boldsymbol{\alpha}}^u}\left( {\theta _{p}^u,\phi _{p}^u} \right)$ with a rotation matrix as
\begin{equation}\label{steer vector rotation}
{{\boldsymbol{\alpha}}^d}\left( {\theta _p^d,\phi _p^d} \right) = \left( {{{\bf{R}}_h}\left( {\theta _p^d,\phi _p^d} \right) \otimes {{\bf{R}}_v}\left( {\theta _p^d} \right)} \right) \cdot 
 {{\boldsymbol{\alpha}}^u}\left( {\theta _{p}^u,\phi _{p}^u} \right),
\end{equation}
where
\begin{equation}\label{horizontal steer vector rotation}
    {{\bf{R}}_h}\left( {\theta _p^d,\phi _p^d} \right) = {\rm{diag}}\left( {\begin{array}{*{20}{c}}
    1\\
    {{e^{j2\pi \frac{{{l_h}\left( {{f^d} - {f^u}} \right)}}{c}\cos \theta _p^d\cos \phi _p^d}}}\\
     \cdots \\
    {{e^{j2\pi \left( {{N_h} - 1} \right)\frac{{{l_h}\left( {{f^d} - {f^u}} \right)}}{c}\cos \theta _p^d\cos \phi _p^d}}}
    \end{array}} \right),
\end{equation}
\begin{equation}\label{vertical steer vector rotation}
    {{\bf{R}}_v}\left( {\theta _p^d} \right) = {\rm{diag}}\left( {\begin{array}{*{20}{c}}
    1\\
    {{e^{j2\pi \frac{{{l_v}\left( {{f^d} - {f^u}} \right)}}{c}\cos \theta _p^d}}}\\
     \cdots \\
    {{e^{j2\pi \left( {{N_v} - 1} \right)\frac{{{l_v}\left( {{f^d} - {f^u}} \right)}}{c}\cos \theta _p^d}}}
    \end{array}} \right),
\end{equation}
are the horizontal rotation matrix and vertical rotation matrix, respectively.

In FDD mode, UL and DL  symbols are transmitted successively in time domain. Fig. \ref{fig_flowchart} demonstrates the flowchart of our proposed framework. The BS utilizes the SRS to extract the channel parameters, and based on the parameters, computes the wideband precoder for DL pilot. According to \cite{3gpp211}, the sounding reference signal (SRS) can be set as cyclical mode with a flexible periodicity ${T_{{\rm{SRS}}}}$ in units of slots. Considering a common configuration where the subcarrier-spacing is 30 ${\rm{kHz}}$, the minimum SRS periodicity ${T_{{\rm{SRS}}}}$ can be as short as 0.5 $\rm{ms}$. We denote the CSI delay by ${T_d} = {N_d}{T_{\rm{SRS}}}$, where $N_d$ is the delay in unit of time slots. The UE computes the complex DL path coefficients upon receiving the DL pilot and feeds them back. The BS finally reconstructs the DL CSI based on the feedback and the extracted parameters. 
\begin{figure}[!t]
\centering
\includegraphics[width=3.2in]{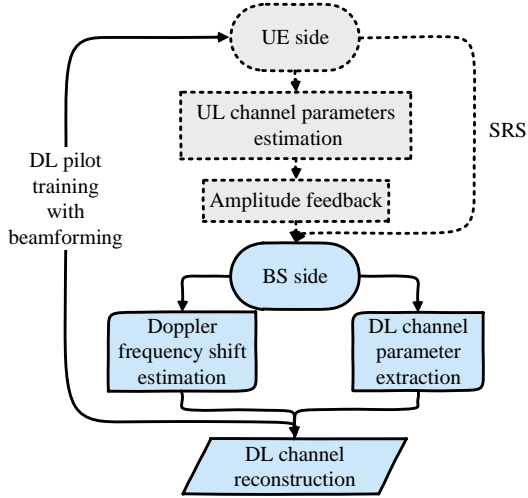}
\caption{Flowchart of our framework.}
\vspace{-0.3cm}
\label{fig_flowchart}
\end{figure}
\section{Uplink channel parameters extraction}\label{sec3}
Our proposed framework depends on the UL channel parameters extraction thanks to the merits of the partial channel reciprocity in FDD. The UL channel parameters is relatively easy to estimate through SRS at the BS side. In order to better exploit the sparsity of the multipath channel, we discuss the UL channel parameters extraction in angle-delay domain.
\subsection{Angle-delay domain projection}
In a wideband massive MIMO system, the UL channel in time domain can be written as
\begin{equation}\label{ul subcarrier channel model}
     {{\bf{h}}^u}\left( {t,f_n^u} \right) = \sum\limits_{p = 1}^P {\beta _p^u{{\bf{\boldsymbol{\alpha} }}^u}\left( {\theta _p^u,\phi _p^u} \right){e^{ - j2\pi f_n^u\tau _p^u}}{e^{jw_p^ut}}} ,
\end{equation}
where $f_n^u ,n \in \left\{ {1, \cdots {N_f}} \right\}$ is frequency of the $n$-th subcarrier. Then the UL channel in matrix form with all subcarriers is
\begin{equation}\label{ul subcarrier all}
    {{\bf{H}}^u}\left( t \right) = \left[ {{{\bf{h}}^u}{{\left( {t,f_1^u} \right)}}, \cdots {{\bf{h}}^u}{{\left( {t,f_{{N_f}}^u} \right)}}} \right].
\end{equation}
The vector form of \eqref{ul subcarrier all} is 
\begin{equation}\label{ul vector channel}
    {{\bf{h}}^u}\left( t \right) = {\bf{vec}}\left( {{{\bf{H}}^u}\left( t \right)} \right) = \sum\limits_{p = 1}^P {\beta _p^u{e^{ - j2\pi {f^u}{\tau^u_p}}}{e^{jw_p^ut}}{{\bf{r}}_p^u}}, 
\end{equation}
where ${{\bf{r}}_p^u} = {\bf{c}}\left( {\tau _p^u} \right) \otimes {{\boldsymbol{\alpha}}^u}\left( {{\theta^u_p},{\phi^u_p}} \right)$ is the angle-delay structure of path $p$ and ${\bf{c}}\left( {\tau _p^u} \right)$ is the delay vector
\begin{equation}\label{delay vector}
    {\bf{c}}\left( {\tau _p^u} \right) = {e^{ - j2\pi {f^u}{f_\Delta }{\tau _p^u}}}{\left[ {1 \cdots {e^{ - j2\pi {\tau _p^u}\left( {{N_f} - 1} \right){f_\Delta }}}} \right]^T}.
\end{equation}
Similarly, the vectorized DL channel is
\begin{equation}\label{dl vector channel}
    {{\bf{h}}^d}\left( t \right) = \sum\limits_{p = 1}^P {\beta _p^d{e^{ - j2\pi {f^d}{\tau^d _p}}}{e^{jw_p^dt}}{\bf{r}}_p^d}.
\end{equation}
A matrix ${\bf{Q}} \in{\mathbb{C}^{{N_t N_f} \times N_t N_f}}$ is used to project ${{\bf{h}}^u}\left( t \right)$ to the angle-delay domain \cite{dft2002,2013BeamSpaceSayeed}
\begin{equation}\label{dft matrix}
    {\bf{Q}}= {\bf{W}}{\left( {{N_f}} \right)^H} \otimes {\bf{W}}\left( {{N_h}} \right) \otimes {\bf{W}}\left( {{N_v}} \right).
\end{equation}
The DFT matrix ${\bf{W}}\left( X \right)$ is calculated by
\begin{equation}\label{dft matrix unit}
{\bf{W}}\left( X \right) = \frac{1}{{\sqrt X }}\left[ {\begin{array}{*{20}{c}}
1&1& \cdots &1\\
1&{{w^{1 \cdot 1}}}& \cdots &{{w^{1 \cdot \left( {X - 1} \right)}}}\\
 \vdots & \vdots & \ddots & \cdots \\
1&{{w^{\left( {X - 1} \right) \cdot 1}}}& \cdots &{{w^{\left( {X - 1} \right)\left( {X - 1} \right)}}}
\end{array}} \right],
\end{equation}
where $w = {e^{\frac{{j2\pi }}{X}}}$. Then the channel in angle-delay domain ${\hat {\bf{g}}^u}\left( t \right) \in{\mathbb{C}^{{N_t N_f} \times 1}}$ is 
\begin{equation}\label{ul angle-delay vector channel}
    {\hat {\bf{g}}^u}\left( t \right)= {{\bf{Q}}^H}{ {\bf{h}}^u}\left( t \right).
\end{equation}
By projecting ${{\bf{h}}^u}\left( t \right)$ to angle-delay domain, we can exploit the channel sparsity and obtain
\begin{equation}\label{ul all adi}
    {{\bf{h}}^u}\left( t \right)= \sum\limits_{i = 1}^{{N_t}{N_f}} {\hat g_i^u\left( t \right){{\bf{q}}_i}},
\end{equation}
where ${{\bf{q}}_i}$ is the $i$-th column of ${\bf{Q}}$ and $\hat g_i^u\left( t \right) = {{\bf{q}}_i}^H{{\bf{h}}^u}\left( t \right)$ is the corresponding complex amplitude. Thanks to the channel sparsity in angle-delay domain, ${{\bf{h}}^u}\left( t \right)$ can be approximated with the linear combination of a relatively small number of selected columns of $\bf{Q}$ which contain most power of the  channel ${{\bf{h}}^u}\left( t \right)$. The set of column indices of ${\bf{Q}}$ is found by
\begin{equation}\label{ul angle-delay path index set}
    {{\cal S}} = \mathop {\arg \min }\limits_{\left|{{\cal S}} \right|} \{ \sum\limits_{l = 1}^{{N_L}} {\sum\limits_{i \in {{\cal S}}} {{{\left| {\hat g_i^u\left( {{t_l}} \right)} \right|}^2}}  \ge \eta } \sum\limits_{l = 1}^{{N_L}} {{{\left| {{{\hat {\bf{g}}}^u}\left( {{t_l}} \right)} \right|}^2}} \},
\end{equation}
 where $\hat g_i^u\left( {{t_l}} \right)$ is the $i$-th row of ${\hat {\bf{g}}^u}\left( {{t_l}} \right)$ and $\eta$ denotes the power threshold. We use $N_L$ channel samples in each UL channel parameter extraction. The size of $\cal{S}$ is denoted by $N_s$ which is referred to as the total number of selected columns in $\bf{Q}$. 

In fact, the index set ${\cal{S}}$ is time-varying and is updated in each UL channel parameter extraction. However, we drop the argument $t$ for simplicity in the rest of the paper. Even though $N_s$ is variant to channel sample and time, we tend to find a fixed $N_s$ satisfying \eqref{ul angle-delay path index set}, which is more convenient to implement in practice. Note that $N_s$ should be carefully chosen not only because it affects the estimation accuracy but also the computation complexity. Thus, there lies a trade-off of $N_s$ between the performance and the complexity. The UL channel can be approximated with $N_s$ angle-delay vectors
\begin{equation}\label{ul angle-delay path selection}
    {\widetilde {\bf{h}}^u}\left( t \right)= \sum\limits_{i \in {\mathcal {S}}} {\hat g_i^u\left( t \right){{\bf{q}}_i}}.
\end{equation}
Comparing \eqref{ul angle-delay path selection} and \eqref{ul vector channel}, the complex amplitude $\hat g_i^u\left( t \right)$ has an implicit physical meaning. Each vector ${\bf{q}}_i$ maps the angle-delay structure ${\bf{r}}_p^u$ and the corresponding $\hat g_i^u\left( t \right)$ maps the complex gain and Doppler frequency $\beta _p^u{e^{ - j2\pi {f^u}{\tau _p}}}{e^{jw_p^ut}}$. An $M$-order superposition of exponentials is utilized to fit the complex amplitude $\hat g_i^u\left( t \right)$ as
\begin{equation}\label{ul angle-delay path amplitude decompose}
    \hat g_i^u\left( t \right)  \buildrel \Delta \over  = \sum\limits_{m = 1}^M {a_m^u} \left( i \right){\left( {z_m^u\left( i \right)} \right)^t},
\end{equation}
where $z_m^u\left( i \right)$ denotes the Doppler frequency and ${a_m^u}\left( i \right)$ denotes the corresponding complex amplitude. Then the UL channel can be approximated by the following form
\begin{equation}\label{ul dft approximate}
    {\widetilde {\bf{h}}^u}\left( t \right) = \sum\limits_{i \in {{\cal S}}} {\sum\limits_{m = 1}^M {a_m^u\left( i \right){\left( {z_m^u\left( i \right)} \right)^t}{{\bf{q}}_i}} } .
\end{equation}
For simplicity, we assume the same $M$ for all selected angle-delay vectors. The value of $M$ should be carefully chosen considering the complexity, Doppler frequency variety and the mismatch problem of DFT projection. More details will be discussed in \secref{sec5}. In the following subsection, we aim to estimate the Doppler frequency $z_m^u\left( i \right)$ with MP method. 
\subsection{Matrix Pencil based Doppler Estimation}
The problem of estimating $z_m^u\left( i \right)$ from channel samples in \eqref{ul angle-delay path amplitude decompose} has a form of a superposition of complex exponentials, where MP method is particularly applicable. {MP method has the advantage of low computation complexity and noise-insensitive \cite{1990MP} over traditional polynomial methods like ESPRIT or Prony.} Therefore, we apply this method in angle-delay domain in order to extract the Doppler. 
We should first introduce Assumption \ref{assumption csi delay time}, which means the stationary time is larger than the CSI delay. 
\begin{assumption}\label{assumption csi delay time}
     During the period of CSI delay $T_d$, channel parameters such as angle and Doppler frequency shift are nearly unchanged.
\end{assumption}
This assumption often holds under a moderate mobility scenario \cite{2019mobilityReportExperi}. Assume the CSI delay is 5 $\rm{ms}$ and the UE speed is 100 $\rm{km/h}$ for example, then the UE moves about 0.14 $\rm{m}$ during this CSI delay period. The position of the UE is approximately unchanged considering that the distance between the UE and the BS is much larger. Therefore, channel parameters such as angles and Doppler barely change during CSI delay period. 

We first briefly introduce the principle of MP method. The Doppler $z_m^u\left( i \right)$ is referred to as the pole in MP. In this method, three parameters are crucial, i.e., sample quantity $N_L$, prediction order $L$ and poles $z_m^u\left( i \right)$. Then, the prediction matrices ${{\bf{P}}_1}\left( i \right),{{\bf{P}}_0}\left( i \right)$ are generated by the complex gain $\hat g_i^u\left( t \right)$ as
\begin{equation}\label{mp prediction second matrix}
{{\bf{P}}_1}\left( i \right)= \left[ {\begin{array}{*{20}{c}}
{\hat g_i^u\left( {{t_{L + 1}}} \right)}&{\hat g_i^u\left( {{t_L}} \right)}& \cdots &{\hat g_i^u\left( {{t_2}} \right)}\\
{\hat g_i^u\left( {{t_{L + 2}}} \right)}&{\hat g_i^u\left( {{t_{L + 1}}} \right)}& \cdots &{\hat g_i^u\left( {{t_3}} \right)}\\
 \vdots & \vdots & \ddots & \vdots \\
{\hat g_i^u\left( {{t_{{N_L}}}} \right)}&{\hat g_i^u\left( {{t_{{N_L} - 1}}} \right)}& \cdots &{\hat g_i^u\left( {{t_{{N_L} - L + 1}}} \right)}
\end{array}} \right],
\nonumber
\end{equation}
\begin{equation}\label{mp prediction first matrix}
    {{\bf{P}}_0}\left( i \right)= \left[ {\begin{array}{*{20}{c}}
{\hat g_i^u\left( {{t_L}} \right)}&{\hat g_i^u\left( {{t_{L{\rm{ - }}1}}} \right)}& \cdots &{\hat g_i^u\left( {{t_1}} \right)}\\
{\hat g_i^u\left( {{t_{L + 1}}} \right)}&{\hat g_i^u\left( {{t_L}} \right)}& \cdots &{\hat g_i^u\left( {{t_2}} \right)}\\
 \vdots & \vdots & \ddots & \vdots \\
{\hat g_i^u\left( {{t_{{N_L}{\rm{ - }}1}}} \right)}&{\hat g_i^u\left( {{t_{{N_L} - 2}}} \right)}& \cdots &{\hat g_i^u\left( {{t_{{N_L} - L}}} \right)}
\end{array}} \right].
\nonumber
\end{equation} 
Drop superscript $u$ for simplicity and construct three matrices
\begin{equation}\label{mp prediction eigenvalue}
    {{\bf{Z}}_0} = {\rm{diag}}\left\{ {{z_1}\left( i \right),{z_2}\left( i \right), \cdots {z_M}\left( i \right)} \right\},
\end{equation}
\begin{small}
\begin{equation}\label{mp prediction first eigenvalue matrix}
{{\mathbf{Z}}_1} = \left[ {\begin{array}{*{20}{c}}
  1&1& \cdots &1 \\ 
  {{z_1}\left( i \right)}&{{z_2}\left( i \right)}& \cdots &{{z_M}\left( i \right)} \\ 
   \vdots & \vdots & \ddots & \vdots  \\ 
  {{z_1}{{\left( i \right)}^{{N_L} - L - 1}}}&{{z_2}{{\left( i \right)}^{{N_L} - L - 1}}}& \cdots &{{z_M}{{\left( i \right)}^{{N_L} - L - 1}}} 
\end{array}} \right],
\nonumber
\end{equation}
\end{small}
\begin{equation}\label{mp prediction second eigenvalue matrix}
{{\mathbf{Z}}_2} = \left[ {\begin{array}{*{20}{c}}
  {{z_1}{{\left( i \right)}^{L - 1}}}&{{z_1}{{\left( i \right)}^{L - 1}}}& \cdots &1 \\ 
  {{z_1}{{\left( i \right)}^{L - 1}}}&{{z_2}{{\left( i \right)}^{L - 2}}}& \cdots &1 \\ 
   \vdots & \vdots & \ddots & \vdots  \\ 
  {{z_M}{{\left( i \right)}^{L - 1}}}&{{z_M}{{\left( i \right)}^{L - 2}}}& \cdots &1 
\end{array}} \right].
\end{equation}
The complex amplitude $a_m^u\left( i \right)$ is given in the form of a diagonal matrix
\begin{equation}\label{mp prediction amplitude matrix}
    {\bf{A}}_u\left( i \right) = {\rm{diag}}\left\{ {{a_1^u}\left( i \right),{a_2^u}\left( i \right), \cdots {a_M^u}\left( i \right)} \right\}.
\end{equation}
The following relationship holds according to \cite{1990MP}
\begin{equation}\label{mp prediction solve first}
    \left\{ \begin{array}{l}
{{\bf{P}}_0}\left( i \right)= {{\bf{Z}}_1}\left( i \right){\bf{A}}_u\left( i \right){{\bf{Z}}_2}\left( i \right),\\
{{\bf{P}}_1}\left( i \right)= {{\bf{Z}}_1}\left( i \right){\bf{A}}_u\left( i \right){{\bf{Z}}_0}\left( i \right){{\bf{Z}}_2}\left( i \right).
\end{array} \right.
\end{equation}

In order to describe the mechanism of how to obtain poles, Lemma \ref{lemma MPP proof} \cite{1990MP} is introduced 
\begin{lemma}\label{lemma MPP proof}
If $M \le L \le N_L - M$, the solution to the singular generalized eigenvalue problem  
\begin{equation}\label{mp eigenvalue}
    \left( {{{\bf{P}}_0}{{\left( i \right)}^\dag }{{\bf{P}}_1}\left( i \right)} \right){\bf{x}} = z{\bf{x}},
\end{equation}
points the way to find poles $z_m^u\left( i \right)$. Each eigenvalue ${z}$ equals to the pole ${z_m^u\left( i \right)}$. $\bf{x}$ is the corresponding eigenvector.
\end{lemma}
After obtaining all poles, the Doppler frequency can be easily calculated.
Algorithm \ref{alg1} explains how to estimate the UL channel parameters like Doppler frequency shifts and the set of angle-delay indices.
\begin{algorithm} [H]                    
\caption{Matrix Pencil based Doppler Estimation}          
\label{alg1}                           
\begin{algorithmic} [1]
\normalsize           
\STATE Initialize $N_L,L,M$, start time $t_s$, end time $t_e$ and obtain channel sample ${\bf{h}}^u\left( t \right)$
\STATE Project ${\bf{h}}^u\left( t \right)$ to angle-delay domain as \eqref{ul all adi}
\STATE Find a suitable $N_s$ satisfying \eqref{ul angle-delay path index set}
\FOR{$t \in \left[ {t_s,t_e} \right]$}
\STATE Obtain the index set $\cal{S}$ 
    \FOR{$n_i\in \left[ {1,N_s} \right]$}
    \STATE Generate prediction matrix ${{\bf{P}}_1}\left( n_i \right)$, ${{\bf{P}}_0}\left( n_i \right)$
    \STATE Using \eqref{mp eigenvalue} to calculate the eigenvalue matrix ${{\bf{Z}}_0}\left( n_i \right)$ 
    \STATE Update $n_i = n_i+1$
    \ENDFOR
 \STATE Update $t = t+1$
\ENDFOR
\STATE Return the UL angle-delay index set and Doppler frequency shift.
\end{algorithmic}
\end{algorithm}
\subsection{Noisy channel sample analysis}
The previous discussion is based on noise-free channel sample assumption. In such cases we can let the prediction order $L=M$ for simplicity. In realistic scenarios, only noisy channel samples are available. In this case, we propose to apply a minimum description length (MDL) criterion \cite{1985mdl} to detect the value $M$ in \eqref{ul dft approximate} and cancel the noise by an $M$ order subtraction Singular Value Decomposition (SVD) where the prediction order satisfies $L>M$. The value of $M$ is minimized under an MDL criterion without prior decision or hypothesis as

\begin{small}
\begin{equation}\label{mdl M value}
    M = \mathop {\min }\limits_{x \in \left\{ {0,1 \cdots L - 1} \right\}} \left\{ {\log {{\left( {\frac{{\prod\limits_{m = x + 1}^L {{z_m}{{\left( i \right)}^{{1 \mathord{\left/
 {\vphantom {1 {\left( {L - x} \right)}}} \right.
 \kern-\nulldelimiterspace} {\left( {L - x} \right)}}}}} }}{{\frac{1}{{L - x}}\sum\limits_{m = x + 1}^L {{z_m}\left( i \right)} }}} \right)}^{ - {N_L}\left( {L - x} \right)}}} \right\},
 \nonumber
\end{equation}
\end{small}
where $z_m\left(i\right)$ is the singular value of
\begin{equation}\label{mdl new prediction matrix}
 {{\bf{P}}_{1,0}} = \left[ {\begin{array}{*{20}{c}}
{{\bf{p}}\left( {{t_{L + 1}}} \right)}&\vline& {{{\bf{P}}_0}}
\end{array}} \right],
\end{equation}
with ${{\bf{p}}\left( {{t_{L + 1}}} \right)}$ being the first column of ${\bf{P}}_1$. After obtaining the value of $M$, the prediction matrix ${{\bf{P}}_{1,0}}$ is calculated after a rank-$M$ truncated SVD
\begin{equation}\label{mdl prediction svd}
    {{\bf{P}}_{1,0}} = {{\bf{U}}_M}{{\bf{\Lambda }}_M}{\bf{V}}_M^H,
\end{equation}
where ${\bf{U}}_M,{{\bf{\Lambda }}_M},{\bf{V}}_M$ are $M$-truncated left singular vector, singular value and right singular vector of ${{\bf{P}}_{1,0}}$, respectively. Then \eqref{mp eigenvalue} in Lemma \ref{lemma MPP proof} becomes
\begin{equation}\label{mdl matrix pencil eigenvalue}
    \left( {{{\bf{P}}_{0,M}}{{\left( i \right)}^\dag }{{\bf{P}}_{1,M}}\left( i \right)} \right){\bf{x}} = z{\bf{x}},
\end{equation}
where 
\begin{equation}\label{mdl truncated prediction matrix}
    \left\{ \begin{array}{l}
{{\bf{P}}_{0,M}}\left( i \right) = {{\bf{U}}_M}{{\bf{\Lambda }}_M}{{\bf{V}}_{M\left( {1:M - 1,:} \right)}}^H,\\
{{\bf{P}}_{1,M}}\left( i \right) = {{\bf{U}}_M}{{\bf{\Lambda }}_M}{{\bf{V}}_{M\left( {2:M,:} \right)}}^H,
\end{array} \right.
\end{equation}
and ${{\bf{V}}_{M\left( {1:M - 1,:} \right)}}$ denotes the sub-matrix consists of the first row to the $\left(M-1\right)$-th row and ${{\bf{V}}_{M\left( {2:M,:} \right)}}$ consists of the second row to the $M$-th row likewise. Finding the eigenvalue of \eqref{mdl matrix pencil eigenvalue} equals to obtaining the poles $z_m\left(i\right)$ in noisy channel sample case.

In this section, the UL channel parameters, such as Doppler frequency shift and angle-delay vector, are obtained at the BS and the UEs. 
These parameters will be used in the following section to facilitate the DL training. 

\section{Downlink Training and Channel Prediction}\label{sec4}
Our CSI acquisition framework relies on channel parameters estimated from the UL channel samples. In Section \ref{sec3}, we have obtained the UL angle-delay vectors ${{\bf{q}}_i}$ and the Doppler frequency $z_m^u\left( i \right)$.  In this section, we introduce the JADD pilot precoding scheme based on the extracted channel parameters and the DL channel reconstruction procedure. 

\subsection{Extract parameters from uplink channel parameters}
Since the UL and DL are operating in different frequency bands, the angle-delay vectors and the Doppler shifts obtained from the UL channel samples have to adapt to the DL frequency band. 
Define the selected UL angle-delay vectors as
 \begin{equation}\label{ul angle-delay vector define}
    {{\mathbf{u}}_j} = \left\{ {{{\mathbf{q}}_i}|i = {i_{{s_j}}},j \in \left\{ {1,2 \cdots {N_s}} \right\}} \right\}, 
 \end{equation}
where the index ${{i_{{s_j}}}}$ denotes the $j$-th index in the UL angle-delay vector index set ${\cal{S}}$.
In order to transform the UL angle-delay vector to the DL one, we introduce Proposition \ref{Propo dl anguar-delay}.
\begin{proposition}\label{Propo dl anguar-delay}
The DL angle-delay vector ${\bf{d}}_j$ is obtained from the UL angle-delay vector ${\bf{u}}_j$ by
\begin{equation}\label{dl angle-delay path Proposition}
{{\bf{d}}_j} = \left( {{{\bf{I}}_{{N_f}}} \otimes {\bf{R}}\left( {\theta _j^d,\phi _{{j}}^d} \right)} \right){{\bf{u}}_j},j \in \left\{ {1,2 \cdots {N_s}} \right\},
\end{equation}
where ${\bf{R}}\left( {\theta _j^d,\phi _{{j}}^d} \right) = {{\bf{R}}_h}\left( {\theta _j^d,\phi _j^d} \right) \otimes {{\bf{R}}_v}\left( {\theta _j^d} \right)$.
\end{proposition}
\begin{proof}
 \quad \emph{Proof:} Please refer to Appendix \ref{Appendix Propo dl angle-delay}. 
\end{proof}
Proposition \ref{Propo dl anguar-delay} demonstrates how to acquire the DL angle-delay vector from the UL ones when the UE is equipped with single antenna. In practice, the UEs may have dual-polarized antennas. The generalization of our method is straightforward, as shown in Remark \ref{remark two polarization}.

\begin{remark}\label{remark two polarization}
    If the UEs are equipped with dual-polarized antennas, the DFT matrix ${\bf{W}}\left( {{N_t}} \right)$ becomes 
    \begin{equation}\label{bipolaration rotatation matrix}
        {\bf{W}}\left( {{N_t}} \right)= \left[ {\begin{array}{*{20}{c}}
    {{\bf{W}}\left( {{N_h}} \right) \otimes {\bf{W}}\left( {{N_v}} \right)}&{}\\
    {}&{{\bf{W}}\left( {{N_h}} \right) \otimes {\bf{W}}\left( {{N_v}} \right)}
    \end{array}} \right].
    \nonumber
    \end{equation}
    Thus, the $j$-th DL angle-delay vector is now
    \begin{equation}\label{dl angluar-delay two polarization}
    {{\bf{d}}_j} = \left( {{{\bf{I}}_{{N_f}}} \otimes \left[ {\begin{array}{*{20}{c}}
    {{\bf{R}}\left( {\theta _j^d,\phi _j^d} \right)}&{}\\
    {}&{{\bf{R}}\left( {\theta _j^d,\phi _j^d} \right)}
    \end{array}} \right]} \right){{\bf{q}}_i}.
    \end{equation}
\end{remark}

Then we calculate the DL Doppler frequency shift with the poles $z_m^u\left( i \right)$ obtained from the UL channel samples
\begin{equation}\label{dl doppler frequency}
   {e^{jw_m^d\left( j \right)}} = {e^{j\frac{{\arccos \left( {\Re \left\{ {\frac{{{z_m^u}\left( i \right)}}{{\left| {{z_m^u}\left( i \right)} \right|}}} \right\}} \right){f^d}}}{{{f^u}}}}}.
\end{equation}
\subsection{DL pilot precoding and CSI reconstruction}
With  the angle-delay vectors and Doppler frequency shifts of the DL channel, we may reconstruct the DL CSI as 
\begin{equation}\label{dl reconstruct channel}
    {\widetilde {\mathbf{h}}^d}\left( t \right) = \sum\limits_{j = 1}^{{N_s}} {\sum\limits_{m{\text{  =  }}1}^M {a_m^d\left( j \right){e^{jw_m^d\left( j \right) {t} }}{{\mathbf{d}}_j}} } ,
\end{equation}
where ${a_m^d}\left( j \right)$ is the $m$-th complex amplitude corresponding to ${{\bf{d}}_j}$. In order to reconstruct the DL channel, ${a_m^d}\left( j \right)$ has to be estimated. We propose to do so with JADD precoded pilot signals. The proposed precoding matrix also helps to reduce the training overhead by exploiting the sparse structure of $\widetilde {\bf{h}}^d\left( t \right)$. 

The vectorized DL channel $\widetilde {\bf{h}}^d\left( t \right)$ can be decomposed to three matrices as 
\begin{equation}\label{dl vector decompose}
    \widetilde {\bf{h}}^d\left( t \right) = {{\bf{D}}}{{\bf{E}}}\left( t \right){{\bf{a}}^d}.
\end{equation}
The DL angle-delay vector matrix ${\bf{D}} \in {\mathbb{C}^{{N_f}{N_t} \times {N_s}}}$ is 
\begin{equation}\label{dl angle-delay path matrix}
{{\bf{D}}} = \left[ {\begin{array}{*{20}{c}}
{{\bf{d}}_1}&{{\bf{d}}_2}& \cdots &{{\bf{d}}_{N_s}}
\end{array}} \right].
\end{equation}
The Doppler matrix ${\bf{E}}\left( t \right) \in {\mathbb{C}^{{N_s} \times {N_s}M}}$ is defined as
\begin{equation}\label{dl doppler matrix}
     {\bf{E}}\left( t \right) = \left[ {\begin{array}{*{20}{c}}
{{{\bf{e}}_1}\left( t \right)}&{}&{}&{}\\
{}&{{{\bf{e}}_2}\left( t \right)}&{}&{}\\
{}&{}& \ddots &{}\\
{}&{}&{}&{{{\bf{e}}_{N_s}}\left( t \right)}
\end{array}} \right],
\end{equation}
where
\begin{equation}\label{dl doppler vector}
{{\mathbf{e}}_j}\left( t \right) = \left[ {\begin{array}{*{20}{c}}
  {{e^{jw_1^d\left( j \right) {t} }}}&{{e^{jw_2^d\left( j \right) {t} }}}& \cdots &{{e^{jw_M^d\left( j \right) {t} }}} 
\end{array}} \right].
\end{equation}
The DL complex amplitude vector ${{\bf{a}}^d} \in {\mathbb{C}^{{{N_s}M}\times 1 }}$ is
\begin{equation}\label{dl complex amplitude matrix}
    {{\bf{a}}^d} = {\left[ {\begin{array}{*{20}{c}}
{{{\bf{a}}^d}\left( 1 \right)}&{{{\bf{a}}^d}\left( 2 \right)}& \cdots &{{{\bf{a}}^d}\left( {{N_s}} \right)}
\end{array}} \right]^T},
\end{equation}
where 
\begin{equation}\label{d complex amplitude vector}
    {{\bf{a}}^d}\left( j \right) = {\left[ {\begin{array}{*{20}{c}}
{a_1^d\left( j \right)}&{a_2^d\left( j \right)}& \cdots &{a_M^d\left( j \right)}
\end{array}} \right]}.
\end{equation}
Using \eqref{dl vector decompose}, we can design a precoding matrix to facilitate DL pilot training. 
This matrix is constructed based on the DL Doppler frequency shifts and angle-delay vectors. 
Traditionally, the idea of precoding is in spatial domain, where the signal is combined in the air from a receiver point of view. However in our scheme, the joint spatial-frequency precoding is a generalized wideband concept. Essentially, the training signal is combined in spatial domain at the BS side, and then combined in frequency domain at the UE side \cite{yin2021codebook}. 
Denote the precoding matrix by ${{\bf{F}}}\left( t \right) \in \mathbb{C}^{N_t N_f \times N_s M}$:
\begin{equation}\label{DL prebeam matrix column}
    {\bf{F}}\left( t \right) = \left[ {\begin{array}{*{20}{c}}
{{{\bf{f}}_1}\left( t \right)}&{{{\bf{f}}_2}\left( t \right)}& \cdots &{{{\bf{f}}_{{N_s}M}}\left( t \right)}
\end{array}} \right].
\end{equation}
Each column of ${\bf{F}}\left( t \right)$, e.g., ${{\mathbf{f}}_n}(t) \in {{\mathbb{C}}^{{N_f}{N_t} \times 1}}$, is composed of the precoding vectors applied on all $N_f$ subcarriers:
\begin{equation}\label{DL prebeaform F per subcarrier}
    {{\mathbf{f}}_n}\left( t \right) = {\left[ {\begin{array}{*{20}{c}}
  {{{\mathbf{f}}_n}{{\left( {t,{f_1}} \right)}^T}}&{{{\mathbf{f}}_n}{{\left( {t,{f_2}} \right)}^T}}& \cdots &{{{\mathbf{f}}_n}{{\left( {t,{f_{{N_f}}}} \right)}^T}} 
\end{array}} \right]^T},
\nonumber
\end{equation}
where 
${{\mathbf{f}}_n}\left( {t,{f_l}} \right) \in {\mathbb{C}^{{N_t} \times 1}}$ is the precoder for the $l$-th subcarrier in the wideband precoder ${{\mathbf{f}}_n}\left( t \right)$. 

Denote the pilot matrix $\bf{S}$ by
\begin{equation}\label{DL pream matrix pilot}
    {\bf{S}}{\rm{ = }}{\left[ {\begin{array}{*{20}{c}}
{{{\bf{s}}_1}^T}&{{{\bf{s}}_2}^T}& \cdots &{{{\bf{s}}_{{N_s}M}}^T}
\end{array}} \right]^T},{{\bf{s}}_n} \in {\mathbb{C}^{1 \times \tau }}.
\end{equation}
where $\tau$ is the length of pilot sequence. Then the transmitted pilot sequence at the  $l$-th subcarrier by the BS is
\begin{equation}\label{dl prebeam matrix cal in space}
 {{\mathbf{g}}^d}\left( {t,{f_l}} \right) = \sum\limits_{n = 1}^{{N_s}M} {{{\mathbf{f}}_n}\left( {t,{f_l}} \right){{\mathbf{s}}_n}}.   
\end{equation}
The received pilot signal by the UE at the  $l$-th subcarrier is 
\begin{equation}\label{dl transmitted signal per subcarrier}
   {{\bf{x}}^d}\left( {t,{f_l}} \right) = {{\bf{h}}^d}{\left( {t,{f_l}} \right)^T}{{\bf{g}}^d}\left( {t,{f_l}} \right) + {\bf{n}}\left( {t,{f_l}} \right),
\end{equation}
where ${{\mathbf{h}}^d}\left( {t,{f_l}} \right)$ denotes the DL channel at the $l$-th subcarrier and ${\bf{n}}\left( {t,{f_l}} \right)$ is the noise at the $l$-th subcarrier. The UE makes a summation over all subcarriers as
\begin{equation}\label{dl prebeam matrix cal in frequency}
    {{\mathbf{y}}^d}\left( t \right) = \sum\limits_{l = 1}^{{N_f}} {{\mathbf{x}^d}\left( {t,{f_l}} \right)}  + {\mathbf{n}}\left( t \right).
\end{equation}
The above-mentioned joint spatial-frequency precoding of the training signal can also be written in matrix form as
\begin{equation}\label{dl prebeam receive}
    {{\bf{y}}^d}\left( t \right) = {{\bf{h}}^d}{\left( t \right)^T}{\bf{F}}\left( t \right){\bf{S}} + {\bf{n}}\left( t \right).
\end{equation}

In the following, we devise our precoding matrix ${\bf{F}}\left( t \right)$. In our framework, the DL channel is reconstructed as \eqref{dl vector decompose}. Thus \eqref{dl prebeam receive} is written as 
\begin{equation}\label{dl prebeamform receive model}
    {\widetilde {\bf{y}}^d}\left( t \right) = \left( {{{\left( {{{\bf{a}}^d}} \right)}^T}{\bf{E}}{{\left( t \right)}^T}{{\bf{D}}^T}{\bf{F}}\left( t \right)} \right){\bf{S}} + {\bf{n}}\left( t \right).
\end{equation}

The Gaussian noise vector ${{\bf{n}}}\left( t \right) \in {\mathbb{C}^{1 \times \tau }}$ has a distribution of ${{\bf{n}}\left(t\right)} \sim {\mathcal {CN}}\left( {0,{\sigma }^2{\bf{I}}} \right)$, where  $\sigma ^2$ is the noise power. Our purpose is to estimate the coefficient vector ${\bf{a}}^d$ and feed it back to the BS.  
Notice that ${\bf{E}}{\left( t \right)^T}{{\bf{D}}^T}$ has a rank of $N_s$ and has no right inverse matrix. Obviously, ${{\bf{E}}}{\left( t \right)^T}$ is of full column rank and ${{\bf{D}}}{^T}$ is of full row rank. Thus, there exists a right inverse matrix of ${{\bf{D}}}{^T}$, however, no right inverse matrix of ${{\bf{E}}}{\left( t \right)^T}$. The Moore-Penrose matrix of ${{\bf{E}}}{\left( t \right)^T}$ is introduced instead and the precoding matrix is designed as
\begin{equation}\label{dl prebeamforming matrix deduce}
    {{\bf{F}}}\left( t \right) = {\left( {{{\bf{D}}}{^T}} \right)^\dag }{\left( {{{\bf{E}}}{{\left( t \right)}^T}} \right)^\dag }.
\end{equation}
Substitute ${{\bf{F}}}\left( t \right)$ with \eqref{dl prebeamforming matrix deduce} and \eqref{dl prebeamform receive model} becomes
\begin{equation}\label{dl prebeamforming result}
    \widetilde {\bf{y}}^d\left( t \right) = \left( {{{{\left( {{{\bf{a}}^d}} \right)}^T}}{{\bf{E}}}{{\left( t \right)}^T}{{\left( {{{\bf{E}}}{{\left( t \right)}^T}} \right)}^\dag }} \right){\bf{S}} + {{\bf{n}}}\left( t \right).
\end{equation}

After applying the precoding matrix, the dimension of ${\bf{S}}$ in \eqref{dl prebeamform receive model} reduces to ${\mathbb{C}^{{N_s}M \times \tau }}$. Due to the channel sparsity in angle-delay domain and a small $M$, the precoding matrix also reduces the training overhead, which does not scale with the number of BS antennas and the bandwidth. 
In principle we should guarantee $\tau  \ge {N_s}M$. For simplicity, the length of the training sequence $\tau$ satisfies $\tau  = {N_s}M$ and $\bf{S}$ is designed as a unitary matrix. Based on \eqref{dl prebeamforming result}, the unknown parameter ${{\bf{a}}^d}$ can be obtained by  least-square (LS) estimation 
\begin{equation}\label{dl parameter estimate}
    {\hat{\bf{a}}^d}={\left( {{\bf{S}}{}^T{{\bf{E}}}{{\left( t \right)}^\dag }{{\bf{E}}}\left( t \right)} \right)^\dag }\widetilde {\bf{y}}^d{\left( t \right)^T}.
\end{equation}
The UEs should feed back the estimated complex coefficient vector ${\hat{\bf{a}}^d}$ to the BS. Therefore, the DL channel after a $T_d$ CSI delay can be easily reconstructed at the BS as
\begin{equation}\label{dl channel reconstruct pred}
    {\widetilde {\mathbf{h}}^d}\left( {t + {T_d}} \right) = \sum\limits_{j = 1}^{{N_s}} {\sum\limits_{m = 1}^M {{\hat{a}}_m^d\left( j \right){e^{jw_m^d\left( j \right)\left( {t + {T_d}} \right)}}{{\mathbf{d}}_j}} }.
\end{equation}
The reconstructed DL channel will be utilized in the downlink precoding for data transmission. 
\section{Performance analysis}\label{sec5}
In our framework, the choice of $N_s$ and $L$ affect the channel prediction performance and the computation complexity. Hence, in this section, we focus on analyzing the impact of $N_s$ and $L$ on the prediction performance, the computational complexity {, and the feedback overhead.}
\subsection{Channel prediction performance analysis}
The DL channel prediction error is defined with the normalized mean square error (NMSE) metric as
\begin{equation}\label{prediction error no quantify}
    \varepsilon  \buildrel \Delta \over = 10\log \mathbb{E}\left\{ {\left\| {\frac{{{{\bf{h}}^d}\left( {t + {T_d}} \right) - {\rm{ }}{{\widetilde {\bf{h}}}^d}\left( {t + {T_d}} \right)}}{{{{\bf{h}}^d}\left( {t + {T_d}} \right)}}} \right\|_2^2} \right\}.
\end{equation}
Define ${t_p} = t + {T_d}$ as the channel prediction offset. We revisit the DL channel reconstruction equation \eqref{dl vector decompose} and substitute ${\bf{a}}^d\left(t\right)$ with \eqref{dl parameter estimate}
\begin{equation}\label{dl vector channel estimation no quantify}
{\widetilde {\bf{h}}^d}\left( {{t_p}} \right) = \widetilde {\bf{h}}_1^d\left( {{t_p}} \right) + \widetilde {\bf{h}}_2^d\left( {{t_p}} \right),
\nonumber
\end{equation}
where  
\begin{equation}\label{dl channel prediction first part}
    {\widetilde {\bf{h}}_1^d}\left( t_p \right) ={{\bf{D}}}{{\bf{E}}}\left( t_p \right){\left( {{{\bf{D}}}{{\bf{E}}}\left( t_p \right)} \right)^\dag }{\bf{h}}^d\left( t_p \right) = {{\bf{D}}}{{\bf{D}}}{^\dag }{\bf{h}}^d\left( t_p \right),
\end{equation}
\begin{equation}\label{dl channel prediction second part}
    {\widetilde {\bf{h}}_{2}^d}\left( t_p \right) = {{\bf{D}}}{{\bf{E}}}\left( t_p \right){\left( {{{\bf{S}}^\dag }} \right)^T}{{\bf{n}}}{\left( t_p \right)^T}.
\end{equation}
Then \eqref{prediction error no quantify} becomes
\begin{equation}\label{dl channel prediction error two part form}
   {\varepsilon } \buildrel \Delta \over = 10\log \mathbb{E}\left\{ {\left\| {\frac{{{\bf{h}}^d\left( t_p \right) - {{\widetilde {\bf{h}}}_{1}^d}\left( t_p \right) - {{\widetilde {\bf{h}}}_{2}^d}\left( t_p \right)}}{{{\bf{h}}^d\left( t_p \right)}}} \right\|_2^2} \right\} . 
\end{equation}
The following theorem gives the lower bound of the DL channel prediction error, which is derived by letting ${N_s}$ take the maximum value, i.e., ${N_s} = {N_f}{N_t}$. 
\begin{theorem}\label{Theorem Ns}
      The lower bound of the DL channel prediction error of the proposed CSI acquisition framework is
      \begin{equation}\label{theorem Ns equation}
         {\varepsilon } = 10\log \left( {\frac{{{\sigma ^2}{N_s}M}}{{\mathbb{E}{{\left\| {{\bf{h}}^d\left( t_p \right)} \right\|}_2}^2}}} \right). 
      \end{equation}
\end{theorem}
\begin{proof}
\quad \emph{Proof:} Please refer to Appendix \ref{Appendix Theroem}.
\end{proof}
Theorem \ref{Theorem Ns} gives the lower bound of the channel prediction error when all the angle-delay vectors in $\mathbf{Q}$ are taken into account. This condition may not be easy to achieve due to the huge feedback overhead and high complexity. 
Fortunately, the sparsity of multipath angles and delays ensures a much smaller $N_s$ in our framework. 
Another important parameter is the prediction order $L$, which may remain small in wideband massive MIMO regime, as shown in the following theorem. First define $N_P$ as the number of non-identical angle-delay structures of all the DL paths. 
\begin{theorem}\label{Thereom M}
     When $L=1,N_L=2,N_s=N_P$, the DL channel prediction error converges to zero as the number of BS antennas and bandwidth increase
     \begin{equation}\label{Theorem M asymptotic}
         \mathop {\lim }\limits_{{N_t},{N_f} \to \infty } \mathbb{E}\left\{ {\left\| {\frac{{{{\bf{h}}^d}\left( t_p \right) - {\rm{ }}{{\widetilde {\bf{h}}}^d}\left( t_p \right)}}{{{{\bf{h}}^d}\left( t_p \right)}}} \right\|_2^2} \right\} = 0.
     \end{equation}
\end{theorem}
\begin{proof}
\quad \emph{Proof:} Please refer to Appendix \ref{Appendix Theorem M}.
\end{proof}
Theorem \ref{Thereom M} gives an asymptotic channel prediction performance of our framework. 
When the number of antennas and bandwidth are finite, we introduce Remark \ref{remark M} for choosing a proper prediction order $L$.
\begin{remark}\label{remark M}
Given any $L$  satisfying $M \le L \le {N_L} - L$, the channel prediction error yields
\begin{equation}
\begin{split}
        &\mathbb{E}\left\{ {\left\| {\frac{{{{\bf{h}}^d}\left( t_p \right) - \widetilde {\bf{h}}_1^d\left( t_p \right)}}{{{{\bf{h}}^d}\left( t_p \right)}}} \right\|_2^2} \right\} - \mathbb{E}\left\{ {\frac{{{N_s}M{\sigma ^2}}}{{\left\| {{{\bf{h}}^d}\left( t_p \right)} \right\|_2^2}}} \right\} \le \varepsilon  \\
        &\le \mathbb{E}{}\left\{ {\left\| {\frac{{{{\bf{h}}^d}\left( t_p \right) - \widetilde {\bf{h}}_1^d\left( t_p \right)}}{{{{\bf{h}}^d}\left( t_p \right)}}} \right\|_2^2} \right\} + \mathbb{E}\left\{ {\frac{{{N_s}M{\sigma ^2}}}{{\left\| {{{\bf{h}}^d}\left( t_p \right)} \right\|_2^2}}} \right\}
\end{split}
\end{equation}
\end{remark}
\begin{proof}
\quad \emph{Proof:} Please refer to Appendix \ref{Appedix remark M}.
\end{proof}
We notice that the difference between the upper bound and lower bound of the channel prediction error is $\mathbb{E}\left\{ {\frac{{{2N_s}M{\sigma ^2}}}{{\left\| {{{\bf{h}}^d}\left( t_p \right)} \right\|_2^2}}} \right\}$, which is scaling with ${1 \mathord{\left/
 {\vphantom {1 {{\text{SNR}}}}} \right.
 \kern-\nulldelimiterspace} {{\text{SNR}}}}$ and is very small when the number of antennas and the bandwith are large. Normally $M$ cannot be known in advance and we assume $M=L$ on noise-free channel sample condition and $M<L$  on noisy channel sample condition, respectively. Therefore, greater $L$ cannot bring significant performance improvement. Remark \ref{remark M} indicates that we should choose as small $L$ as possible for a given $M$ and $N_L$. This observation is also confirmed in simulation of Sec. \ref{sec6}. 
 
 However, the limited number of antennas causes DFT mismatch problem. Thus, each ${\bf{d}}_j$ cannot accurately map the exact angle-delay structure of the DL channel. Therefore, each angle-delay vector may correspond to multiple Doppler frequency shifts. Bigger $L$ may better fit the corresponding Doppler frequency of each angle-delay vector. Thus, there lies a trade-off in the choice of $L$. Since the diversity of Doppler frequency shift cannot be known apriori in realistic applications, the optimal $L$ is difficult to obtain. Thus a relatively small $L$ satisfying $M \le L \le {N_L} - L$ is recommended.
\subsection{Complexity and feedback overhead analysis}
Our DL channel reconstruction framework consists of five parts, i.e., the DFT projection, the MP based Doppler estimation, the UL to DL transformation of angle-delay vectors and Dopplers, the DL training, and channel reconstruction. The DFT projection can be realized with fast Fourier transform (FFT), which has a complexity of $\mathcal{O}\left( {{N_f}{N_t}{{\log }_2}\left( {{N_f}{N_t}} \right)} \right)$. The complexity of MP method is mainly the SVD, i.e., $\mathcal{O}\left( {{{\left( {{N_L} - L} \right)}^2}L + \left( {{N_L} - L} \right){L^2}} \right)$. The complexity of parameter transformation procedure is $\mathcal{O}\left( {{N_s}{N_v} + {N_s}{N_h}} \right) + \mathcal{O}\left( {{N_f}^2{N_t}^2} \right) + \mathcal{O}\left( {{N_s}M{{\log }^2}\left( {{N_s}M} \right)} \right) + \mathcal{O}\left( {{N_f}{N_t}{N_s}^2{M^2}} \right)$. The DL training contains a matrix inversion and the complexity is $\mathcal{O}\left( { {{N_f}{N_t}{N_s}M}} \right)$. The channel reconstruction entails a matrix inversion and SVD which have the complexity of $\mathcal{O}\left( {{N_s}^3M + {N_s}^3{M^2}} \right) + \mathcal{O}\left( {{N_s}^3{M^2}} \right) + \mathcal{O}\left( {{N_s}^2{M^2}} \right) + \mathcal{O}\left( {{{\left( {{N_s}M} \right)}^{2.37}}} \right)$. {The overall complexity of the channel reconstruction procedure is thus ${\mathcal{O}}\left( {{N_f}{N_t}{N_s}^2{M^2}} \right)$. Obviously, our framework is of polynomial complexity and requires no iterative computing like CS methods or machine learning methods.}

{The feedback overhead is now analyzed for a given channel coherence time $T_{c}$. In Enhanced Type II codebook \cite{3gpp214}, the feedback overhead scales with $L_1L_2$, where $L_1,L_2$ are smaller than $N_t,N_f$. In classical CS methods like \cite{2015CS}, the feedback overhead depends on the reduced dimension $N_r$ of the channel and scales with $N_rN_f$. In other methods like NOMP \cite{2019hanEfficient} and deep learning \cite{2020DeepHan}, the feedback overhead depends on the number of paths $L_p$, which is large in rich scattering environments. And for a wideband system, the feedback overhead of these methods scales with $L_pN_f$. 
Thanks to the channel prediction capability, the feedback overhead of our framework is ${{{N_s}M} \mathord{\left/
 {\vphantom {{{N_s}M} {{N_{c}}}}} \right. 
 \kern-\nulldelimiterspace} {{N_{c}}}}$ scalars for one channel coherence time $T_{c}$, where ${N_{c}} \ge 1$ means only one set of ${N_s}M$ feedback coefficients is required for a time interval of ${N_{c}T_c}$. Hence, our framework has the advantage of reduced feedback over these traditional methods.}

\section{Numerical results}\label{sec6}
 In this section, we validate the proposed JADD framework with the industrial channel model of the cluster-delay-line-A (CDL-A) defined by 3GPP \cite{3gpp901} in a rich scattering scenario. Unless particularly specified, CDL-A channel model contains a total of 23 clusters with 20 paths inside each cluster. Following the $\rm{n}65$ new radio (NR) band in \cite{3gpp104}, the UL center frequency is 1.92 $\rm{GHz}$ and the DL center frequency is 2.11 $\rm{GHz}$. The bandwidth of UL and DL are both 20 $\rm{MHz}$ with a 30 $\rm{kHz}$ subcarrier spacing, implying that 51 resource blocks (RBs) are available per time slot. In this configuration, each time slot contains 14 OFDM symbols and is as short as 0.5 $\rm{ms}$, which denotes the SRS signal cycle length. The BS antenna configuration is $\left( {{N_v},{N_h},{P_t}} \right){\rm{ = }}\left( {2,8,2} \right)$, where  ${P_t}$ is the number of polarizations for each antenna element. The spacing between the antenns in vertical direction and horizontal direction are both 0.5${c \mathord{\left/
 {\vphantom {c {{f^d}}}} \right.\kern-\nulldelimiterspace} {{f^d}}}$. Table \ref{tab simulation config} gives all the other parameters used in our simulation unless otherwise specified. The DL precoding process is the Eigen Zero Forcing (EZF) \cite{2010TSPeigenvalue} and the UEs apply Minimum Mean Square Error-Interference Rejection Combining (MMSE-IRC) receiver. The performance of our framework is shown in two metrics, the SE and the prediction error (PE). The PE is defined by \eqref{prediction error no quantify}. {The spectral efficiency ${R_s}$ is calculated over a period of time and all subcarriers by 
 \begin{equation}
     {R_s}{\rm{  =  }}\mathbb{E}\left\{ {\sum\limits_{k = 1}^K {{\rm{log}}\left( {1 + \frac{{\left\| {\overline {\bf{h}} _k^d\left( {t,f} \right){{\bf{G}}_k}\left( {t,f} \right)} \right\|_2^2}}{{\sigma _k^2 + \sum\limits_{j \ne k}^K {\left\| {\overline {\bf{h}} _j^d\left( {t,f} \right){{\bf{G}}_j}\left( {t,f} \right)} \right\|} }}} \right)} } \right\},
     \nonumber
 \end{equation}
 where ${{\mathbf{G}}_k}\left( {t,f} \right)$ is the precoding matrix and $\overline {\mathbf{h}} _k^d\left( {t,f} \right)$ is the estimated channel. ${\sigma _k^2}$ is the noise power at UE $k$. }
\begin{table}[!t]
\centering \protect\protect\caption{System Parameters in Simulations}
\label{tab simulation config}
\begin{tabular}{|c|c|}
\hline Physical meaning & Default value\tabularnewline
\hline
Channel model & CDL-A\tabularnewline
\hline
Bandwidth & 20 $\rm{MHz}$\tabularnewline
\hline
UL carrier frequency & 1.92 $\rm{GHz}$\tabularnewline
\hline
DL carrier frequency & 2.11 $\rm{GHz}$\tabularnewline
\hline
Subcarrier spacing & 30 $\rm{kHz}$\tabularnewline
\hline
Resource block & 51\tabularnewline
\hline
Angle spread RMS & $\left( {{{87.1}^\circ },{{33.6}^\circ },{{102.1}^\circ },{{24.7}^\circ }} \right)$\tabularnewline
\hline
Delay spread & 300 $\rm{ns}$\tabularnewline
\hline
Number of paths & 460\tabularnewline
\hline
\tabincell{c}{Transmit antenna \\configuration} & \tabincell{c}{$\left( {{N_v},{N_h},{P_t}} \right)= \left( {2,8,2} \right)$, \\ polarization direction are ${0^\circ },{90^\circ }$}\tabularnewline
\hline
\tabincell{c}{Receive antenna \\configuration} & \tabincell{c}{$\left( {{N_v},{N_h},{P_t}} \right)= \left( {1,1,2} \right)$, \\ polarization direction are $\pm {45^\circ }$}\tabularnewline
\hline
Slot duration & 0.5 $\rm{ms}$\tabularnewline
\hline
Number of UEs & 8\tabularnewline
\hline
\end{tabular}
\end{table}

{Three baseline schemes are introduced as the benchmarks. All baselines follow the same channel parameters in Table \ref{tab simulation config}. The curves labeled with ``Enhanced Type II with perfect CSI" are the performances of Enhanced Type II codebook where perfect CSI is known by the UEs and there is no CSI delay.  The first baseline is Enhanced Type II codebook, yet only delayed CSI is known by the UEs.} The second baseline is utilzing an adaptive and parameter free recurrent neural structure (APF-RNS) based on deep learning \cite{2019MLcompare} for real-time prediction. In this method, the DL channel is predicted according to the temporal correlation with the recent history DL channel data with CSI delay, however, without any CSI compression or quantization. The APF-RNS network structure follows the configuration in \cite{2019MLcompare}, and adopts 32 long-short term memory (LSTM) units as the hidden layer. In addition, the history channel data for training and testing is generated by the same channel model with parameters of Table \ref{tab simulation config}. We perform the online training of APF-RNS with the known length of 20 and the prediction length of 1. {The last baseline is a traditional CS method, called TVAL3 \cite{2013TVAL3}. We apply this method to compress the dimension of the DL channel and recover it through a TVAL3 solver. The basic parameter setting is following \cite{2009TVAL3} and the compression ratio is $\rho=1/4$.} Our proposed scheme is referred to as JADD scheme in our simulations.

\begin{figure}[ht]
\centering
\includegraphics[width=3.2in]{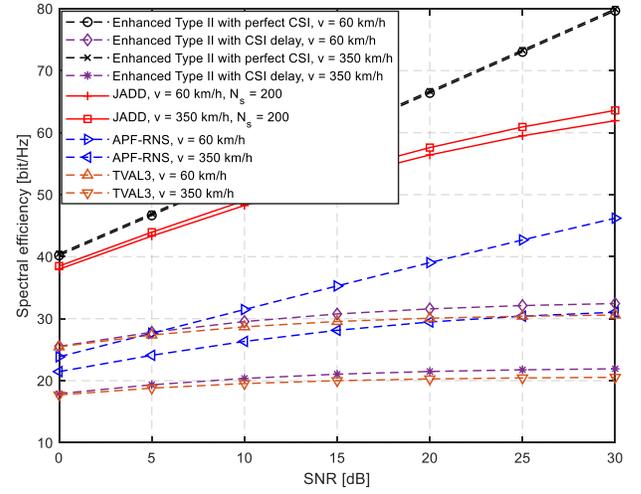}
\caption{{SE performance vs. SNR, different UE speeds, noise-free channel samples,  $L=2,T_d=5$ $\rm{ms}$.}}
\vspace{-0.3cm}
\label{figure_speed_se}
\end{figure}
{In Fig. \ref{figure_speed_se}, our framework is evaluated with the SE metric. In both high speed ($350$ km/h) and low speed ($60$ km/h), our JADD outperforms Enhanced Type II, APF-RNS and TVAL3, which demonstrate the superiority of our scheme in different mobility scenarios. We also conclude that CSI delay causes sever performance dropping in all baselines, especially in high mobility scenario. Note that there exists a small SE difference between the cases when the speed of UE is $60$ km/h and $350$ km/h. This phenomenon is caused by the different angle-delay sparsity and the corresponding Doppler frequency of the channel under different mobility scenarios. }
 
\begin{figure}[ht]
\centering
\includegraphics[width=3.2in]{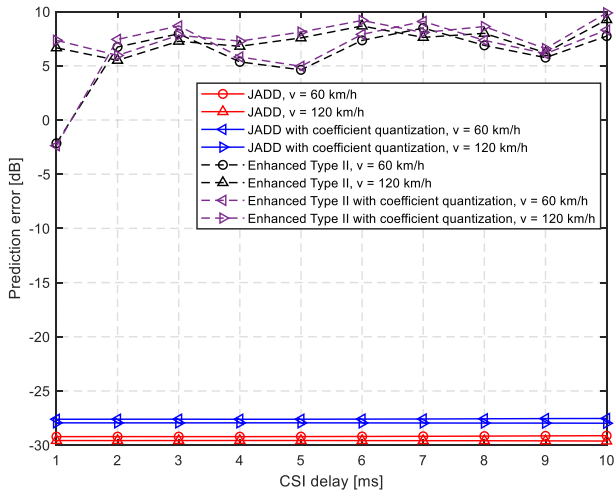}
\caption{{PE performance vs. CSI delay, noise-free channel samples, feedback coefficients quantized, $N_s=200,L=2$, ${\cal{C}}_a=4$, ${\cal{C}}_p=6$.}}
\vspace{-0.3cm}
\label{figure_delay_pe}
\end{figure}

 {Fig. \ref{figure_delay_pe} shows the PE performance under different CSI delay and different mobility levels. Following the coefficient quantization method in 5G \cite{3gpp214}, we use a ${\cal{C}}_a$-bit geometric sequence codebook and a ${\cal{C}}_p$-bit geometric sequence codebook to quantize the amplitude and the phase of the feedback coefficients, respectively.}

{Because of the CSI delay, the PE of Enhanced Type II is unsatisfactory, especially in high mobility scenarios. Compared with Enhanced Type II, our JADD can overcome the performance degradation brought by the CSI delay since our novel channel reconstruction framework \eqref{dl channel reconstruct pred}} can well predict the DL channel. We may conclude that our scheme can well adapt to different CSI delay even with  quantization errors of feedback coefficients. 

\begin{figure}[ht]
\centering
\includegraphics[width=3.2in]{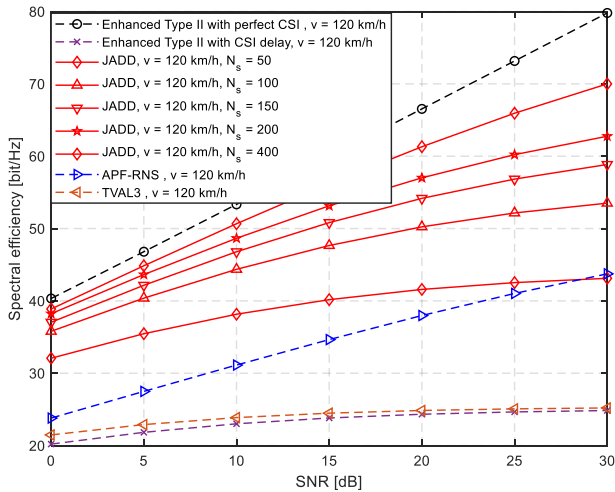}
\caption{{SE performance vs. SNR, different $N_s$, noise-free channel samples, $L=2$, $T_d=5$ $\rm{ms}$.}}
\vspace{-0.3cm}
\label{figure_ns_se}
\end{figure}
Fig. \ref{figure_ns_se} demonstrates the SE performance under different $N_s$ values while the performance of our framework always surpasses {the three baselines even with a small $N_s$. The SE of JADD quickly increases with $N_s$.}  

\begin{figure}[ht]
\centering
\includegraphics[width=3.2in]{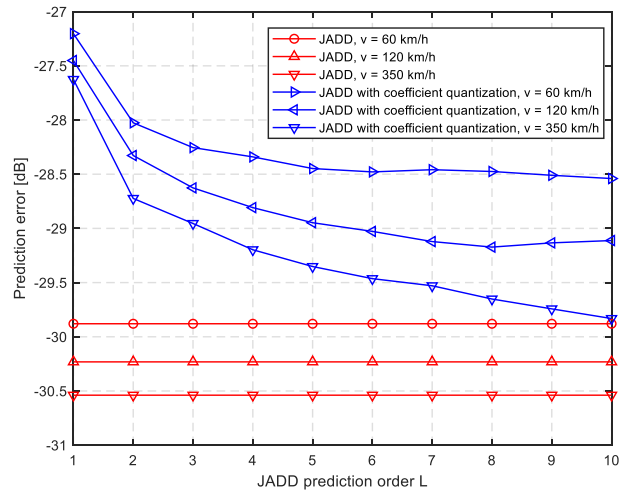}
\caption{{PE performance vs. prediction order $L$, noise-free channel samples, feedback coefficients quantized, $N_s=200,T_d=5$ $\rm{ms}$, ${\cal{C}}_a=4$, ${\cal{C}}_p=6$.}}
\vspace{-0.3cm}
\label{figure_predorder_pe}
\end{figure}

 {Fig. \ref{figure_predorder_pe} shows the PE performance of our method under different prediction order $L$. The results show that the value} of $L$ has little impact on the PE performance, which is aligned with Remark \ref{remark M}. We notice that the PE with coefficient quantizations varies a little with the value of $L$. {This phenomenon is reasonable. The quantization error results in the failure of Remark \ref{remark M} and bigger $L$ leads to higher Doppler frequency resolution, hence better PE performance.} However, in Fig. \ref{figure_predorder_pe}, bigger $L$ leads to only a small improvement of the PE with quantization but much heavier computation complexity. Therefore, in realistic application, a small $L$ is still preferable considering the trade-off between the complexity and prediction error.  
\begin{figure}[ht]
\centering
\includegraphics[width=3.2in]{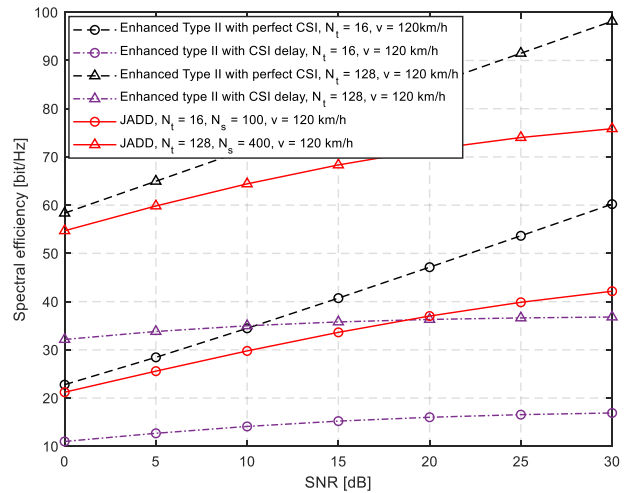}
\caption{SE performance vs. SNR, different BS antenna configurations, noise-free channel samples, $L=2$, $\eta = 0.99, T_d=5$ $\rm{ms}$.}
\vspace{-0.3cm}
\label{figure_antenna_se}
\end{figure}

{In the follow, we focus on evaluating the robustness of our method in the cases of different antenna configurations, different channel models and inaccurate CSI samples. }

The previous numerical results are based on the same antenna configuration and Fig. \ref{figure_antenna_se} shows the SE performance of JADD under different BS antenna configurations. Note that the value of $N_s$ should be carefully chosen to assure \eqref{ul angle-delay path index set} under different $N_t$ configuration given the threshold $\eta$. 
The results show that the SE performance of our framework always outperforms the Enhanced Type II codebook with CSI delay.

\begin{figure}[ht]
\centering
\includegraphics[width=3.2in]{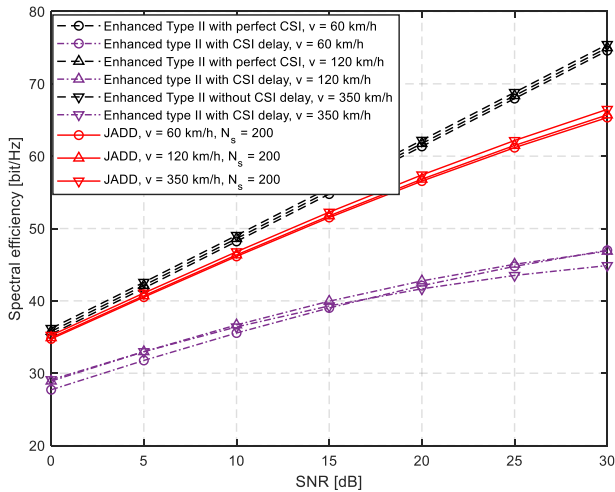}
\caption{SE performance under CDL-D channel model, noise-free channel samples, $L=2,T_d=5$ $\rm{ms}$.}
\vspace{-0.3cm}
\label{figure_cdlD_se}
\end{figure}
In a rich scattering environment, our framework performs well as discussed above. In fact, different scattering environments render different physical features of the channel. The channel model of CDL-D \cite{3gpp901} which contains a line of sight (LOS) path is also considered. The numerical result is demonstrated in Fig. \ref{figure_cdlD_se}. Our framework still performs well in this case. We notice that the same $N_s$ in CDL-D achieves better SE performance as in CDL-A. This phenomenon is reasonable because CDL-D channel has higher angle-delay sparsity.

\begin{figure}[ht]
\centering
\includegraphics[width=3.2in]{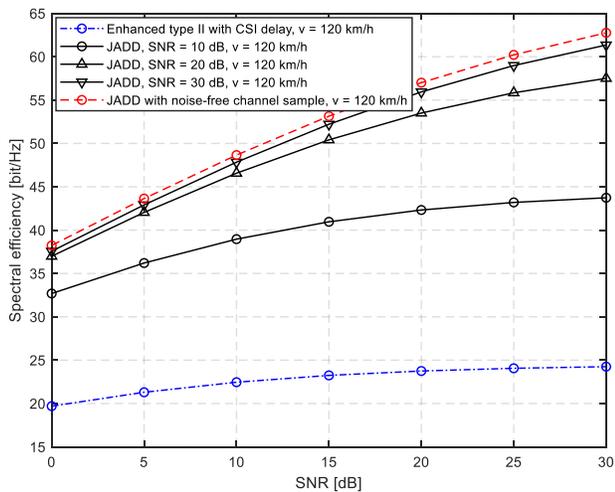}
\caption{SE performance vs. SNR, noisy channel samples with sampling SNR= $\{$10 dB, 20 dB, 30 dB$\}$, $N_s=200,L=2,T_d=5$ $\rm{ms}$.}
\vspace{-0.3cm}
\label{figure_noise_se}
\end{figure}

So far, the previous numerical results are achieved under the noise-free channel sample condition. Now we show the performance of JADD under noisy channel sample case. The channel sample noise is modeled as i.i.d. Gaussian and the noise power is characterized by channel sample SNR. Fig. \ref{figure_noise_se} shows the SE performance of our framework under different channel sample SNRs. Our framework outperforms the Enhanced Type II codebook with CSI delay. 
Therefore, our scheme is robust to noisy channel samples.

\section{Conclusion}\label{sec7}
In this paper we proposed a novel channel prediction framework to address the curse of mobility in FDD massive MIMO, which suffers from the problems of both the CSI aging and large training overhead. 
Our framework combined the merits of partial channel reciprocity in FDD and the angle-delay-Doppler structure of the multipath channel. The DL channel was reconstructed with parameters extracted from the UL channel parameters and some coefficients fed back from the UEs. In particular, the BS calculates the angle-delay vectors of the UL channel and estimates the Doppler frequency shifts using the MP method. A wideband JADD precoding matrix was then proposed to facilitate the acquisition of the desired coefficients, and meanwhile, reduce the training overhead. At the UE side, only some scalar coefficients were computed and fed back to the BS. Our asymptotic analysis showed the prediction error converges to zero as the number of BS antennas increases while only two UL channel samples are available. A scheme to choose a proper prediction order was also discussed. The numerical results demonstrated that our framework works well even in high-mobility scenarios with large CSI delays.

\appendices
\section{Proof of Proposition \ref{Propo dl anguar-delay}}\label{Appendix Propo dl angle-delay}
\begin{proof}
\emph{Proof:} First we need to derive the relationship between $i$ and $j$. The $i$-th column of DFT matrix is 
\begin{equation}\label{angle-delay structure prove first}
{{\bf{q}}_i} = {\left[ {\begin{array}{*{20}{c}}
{{{\bf{w}}_i}}&{{e^{ - j\frac{{2\pi }}{{{N_f}}}i}}{{\bf{w}}_i}}& \cdots &{{e^{ - j\frac{{2\pi }}{{{N_f}}}\left( {{N_f} - 1} \right)i}}{{\bf{w}}_i}}
\end{array}} \right]^T},
\end{equation}
where ${{\bf{w}}_i} = {\left[ {\begin{array}{*{20}{c}}
1&{{e^{ - j\frac{{2\pi }}{{{N_t}}}i}}}& \cdots &{{e^{ - j\frac{{2\pi }}{{{N_t}}}i\left( {{N_t} - 1} \right)}}}
\end{array}} \right]^T}$. Comparing ${{\bf{q}}_i}$ in \eqref{angle-delay structure prove first} and ${{\bf{r}}_p}$ in \eqref{ul vector channel}, angle information $\theta _p^u,\phi _p^u$ and delay information $\tau _p^u$ are closely related to elements in ${{\bf{q}}_i}$. Each angle-delay vector index $i$ corresponds to an index combination $\left( {{i^a},{i^h},{i^v}} \right)$, where ${i^a}$ denotes the angle index, ${i^h}$ denotes the horizontal angle index and ${i^v}$ denotes the vertical angle index. The relationship between $i^a$ and $i$ is
\begin{equation}\label{angle-delay path index change}
    {i^a} = \left\{ \begin{array}{l}
\bmod \frac{i}{{{N_t}}},\bmod \frac{i}{{{N_t}}} \ne 0,\\
{N_t},\bmod \frac{i}{{{N_t}}} = 0.
\end{array} \right.
\end{equation}
Then, the horizontal and the vertical angle index are calculated as
\begin{equation}\label{angle-delay index change horizontal and vertical component}
{i^h} = \left\{ {\begin{array}{*{20}{l}}
{\frac{{\left( {{i^a} - {i^v}} \right)}}{{{N_v}}} + 1,{i^v} \ne 0,}\\
{\frac{{{i^a}}}{{{N_v}}},{i^v} = 0,}
\end{array}} \right.
\end{equation}

\begin{equation}\label{angle-delay index change vertical component}
{i^v} = \left\{ {\begin{array}{*{20}{l}}
{\bmod \frac{{{i^a}}}{{{N_v}}},\bmod \frac{{{i^a}}}{{{N_v}}} \ne 0,}\\
{{N_v},\bmod \frac{{{i^a}}}{{{N_v}}} = 0.}
\end{array}} \right.
\end{equation}
Thus, each UL angle-delay  $i$ map a DL channel angular index $j,j\in\left\{{1,2\cdots,N_s}\right\}$ as
\begin{equation}\label{dl angle index}
\sin {\theta _j}\sin {\phi _j} = \frac{{{i^h}c}}{{{l_h}{f^u}{N_h}}},\cos {\theta _j} = \frac{{c{i^v}}}{{{l_v}{f^u}{N_v}}}.
\end{equation}
According to \eqref{steer vector rotation}, the DL angle-delay structure can be calculated as
\begin{small}
\begin{equation}\label{dl angle-delay path}
\begin{split}
&{{\boldsymbol{r}}^d}\left( {{\theta _j},{\phi _j}} \right) = \left( {{{\bf{I}}_{{N_f}}}{\bf{c}}\left( {\tau _i^u} \right)} \right) \otimes \left( {\left( {{{\bf{R}}_h}\left( {{\theta _j},{\phi _j}} \right) \otimes {{\bf{R}}_v}\left( {{\theta _j}} \right)} \right){{\boldsymbol{\alpha}}^u}\left( {{\theta _i},{\phi _i}} \right)} \right)\\
 &= \left( {{{\bf{I}}_{{N_f}}} \otimes \left( {{{\bf{R}}_h}\left( {{\theta _j},{\phi _j}} \right) \otimes {{\bf{R}}_v}\left( {{\theta _j}} \right)} \right)} \right) \cdot \left( {{\bf{c}}\left( {\tau _i^u} \right) \otimes {{\boldsymbol{\alpha}}^u}\left( {{\theta _i},{\phi _i}} \right)} \right).
\end{split}
\nonumber
\end{equation}
\end{small}
We use the angle-delay vector to approximate the DL angle-delay structure
\begin{equation}\label{dl angle-delay path approximate}
    {\bf{d}}_j \buildrel \Delta \over = \left( {{{\bf{I}}_{{N_f}}} \otimes \left( {{{\bf{R}}_h}\left( {{\theta _{{j}}},{\phi _{{j}}}} \right) \otimes {{\bf{R}}_v}\left( {{\theta _{{j}}}} \right)} \right)} \right){{\bf{u}}_j}.
\end{equation}
Then the Proposition \ref{Propo dl anguar-delay} is proved.
\end{proof}
\section{Proof of Theorem \ref{Theorem Ns}}\label{Appendix Theroem}
\begin{proof}
\emph{Proof:} The vectorized DL channel can be estimated by a series of angle-delay vectors superposition mapped by the columns of DFT matrix $\bf{Q}$. As $N_s$ increases, angle-delay resolution improves. The power leakage problem alleviates, thus, ${\varepsilon }$ decreases. 

When ${N_s} = {N_f}{N_t}$, ${{\bf{D}}}{{\bf{D}}}^\dag = {{\bf{I}}_{{N_f}{N_t}}}$ holds. The prediction error \eqref{prediction error no quantify} becomes
\begin{equation}\label{appendxi dl prediction deduce form for part two}
\begin{split}
  \varepsilon & = 10\log \left( {\frac{{{\mathbb{E}}\left\| {{{\mathbf{h}}^d}\left( t_p \right) - {{\mathbf{h}}^d}\left( t_p \right){\mathbf{I}} - \widetilde {\mathbf{h}}_2^d\left( t_p \right)} \right\|_2^2}}{{{\mathbb{E}}\left\| {{{\mathbf{h}}^d}\left( t_p \right)} \right\|_2^2}}} \right)\hfill \\
   &= 10\log \left( {\frac{{{\mathbb{E}}\left\{ {\widetilde {\mathbf{h}}_2^d\left( t_p \right)\widetilde {\mathbf{h}}_2^d{{\left( t_p \right)}^H}} \right\}}}{{{\mathbb{E}}\left\| {{{\mathbf{h}}^d}\left( t_p \right)} \right\|_2^2}}} \right) \hfill \\
  & = 10\log \left( {\frac{{{\mathbb{E}}\left\| {{\mathbf{DE}}\left( t_p \right){{\left( {{{\mathbf{S}}^\dag }} \right)}^T}} \right\|_2^2{\mathbb{E}}\left\| {{\mathbf{n}}{{\left( t_p \right)}^T}} \right\|_2^2}}{{{\mathbb{E}}\left\| {{{\mathbf{h}}^d}\left( t_p \right)} \right\|_2^2}}} \right). \hfill \\ 
\end{split}
\end{equation}
The training sequence matrix is denoted in column vectors and row vectors form
\begin{equation}\label{defination for S}
\begin{split}
{\left( {{{\bf{S}}^\dag }} \right)^T} &= {\left[ {\begin{array}{*{20}{c}}
{{{\bf{s}}_{1,1}}}&{{{\bf{s}}_{1,2}}}& \cdots &{{{\bf{s}}_{{N_s},M}}}
\end{array}} \right]^T}\\
 &= \left[ {\begin{array}{*{20}{c}}
{{\bf{s}}\left( 1 \right)}&{{\bf{s}}\left( 2 \right)}& \cdots &{{\bf{s}}\left( {{N_s}M} \right)}
\end{array}} \right],
\end{split}
\end{equation}
where ${{\bf{s}}_{j,m}} = \left[ {\begin{array}{*{20}{c}}
{{s_{j,m}}\left( 1 \right)}&{{s_{j,m}}\left( 2 \right)}& \cdots &{{s_{j,m}}\left( {{N_s}M} \right)}
\end{array}} \right]$ and \\${\bf{s}}\left( n \right) = {\left[ {\begin{array}{*{20}{c}}
{{s_{1,1}}\left( n \right)}&{{s_{1,2}}\left( n \right)}& \cdots &{{s_{{N_s},M}}\left( n \right)}
\end{array}} \right]^T}$. Calculate the matrix
\begin{equation}\label{structure of matrix C*E*S}
\begin{array}{l}
{{\bf{D}}}{{\bf{E}}}\left( t_p \right){\left( {{{\bf{S}}^\dag }} \right)^T}
{ = }{\left[ {\begin{array}{*{20}{c}}
{\sum\limits_{j = 1}^{{N_s}} {{\bf{d}}_j\sum\limits_{m = 1}^M {{e^{jw_{j,m}^d\left( t_p \right)}}{s_{j,m}}\left( 1 \right)} } }\\
{\sum\limits_{j = 1}^{{N_s}} {{\bf{d}}_j\sum\limits_{m = 1}^M {{e^{jw_{j,m}^d\left( t_p \right)}}{s_{j,m}}\left( 2 \right)} } }\\
 \cdots \\
{\sum\limits_{j = 1}^{{N_s}} {{\bf{d}}_j\sum\limits_{m = 1}^M {{e^{jw_{j,m}^d\left( t_p \right)}}{s_{j,m}}\left( {{N_s}M} \right)} } }
\end{array}} \right]^T}
\end{array}.
\end{equation}
Let $\sum\limits_{m = 1}^M {{e^{jw_{j,m}^d\left( t_p \right)}}{s_{j,m}}\left( n \right)}  = {\lambda _{j,n}}\left( t_p \right)$ and obtain 
\begin{equation}\label{lemma digonal C*E*S structure}
    \begin{split}
        &{\left\| {{\bf{DE}}\left( t_p \right){{\left( {{\bf{S}}{{\left( t_p \right)}^\dag }} \right)}^T}} \right\|_2^2}\\ &= \sum\limits_{n = 1}^{{N_s}M} {\left( {\left( {\sum\limits_{j = 1}^{{N_s}} {{\lambda _{j,n}}{{\left( t_p \right)}^H}{\bf{d}}_j^H} } \right)\left( {\sum\limits_{j = 1}^{{N_s}} {{\lambda _{j,n}}\left( t_p \right){{\bf{d}}_j}} } \right)} \right)}.
    \end{split}
\end{equation}
The following lemma shows that ${\bf{d}}_j$ has unit norm and mutual orthogonality. 
\begin{lemma}\label{lemma digonal}
        For any $i,j  = 1, \cdots, N_s$, we have 
        \begin{equation}\label{lemma digonal result}
            {\bf{d}}_i{^H}{\bf{d}}_j = \left\{ \begin{array}{l}
            1,i = j\\
            0,i \ne j
            \end{array} \right.. 
        \end{equation}
    \end{lemma}
    \begin{proof}
 \quad \emph{Proof:} Eq. \eqref{lemma digonal result} is derived as
        \begin{equation}\label{lemma 3 ci*cj}
            \begin{split}
            &{\bf{d}}_i^H{\bf{d}}_j
            = {{\bf{q}}_i}^H{\left( {\left( {{{\bf{I}}_{{N_f}}} \otimes \left( {{{\bf{R}}_h}\left( {{\theta _{{i}}},{\phi _i}} \right) \otimes {{\bf{R}}_v}\left( {{\theta _{{i}}}} \right)} \right)} \right)} \right)^H}\\
             &\cdot \left( {{{\bf{I}}_{{N_f}}} \otimes \left( {{{\bf{R}}_h}\left( {{\theta _{{j}}},{\phi _j}} \right) \otimes {{\bf{R}}_v}\left( {{\theta _{{j}}}} \right)} \right)} \right){{\bf{q}}_j}\\
              &= {{\bf{q}}_i}^H\left( {{{\bf{I}}_{{N_f}}} \otimes {{\bf{R}}_{i,j}}} \right){{\bf{q}}_j},
            \end{split}
        \end{equation}
        where 
        \begin{equation}\label{R_i_j equation}
        {{\bf{R}}_{i,j}} = \left( {{{\bf{R}}_h}{{\left( {{\theta _{{i}}},{\phi _i}} \right)}^H}{{\bf{R}}_h}\left( {{\theta _{{j}}},{\phi _j}} \right)} \right)\otimes \left( {{{\bf{R}}_v}{{\left( {{\theta _{{i}}}} \right)}^H}{{\bf{R}}_v}\left( {{\theta _{{j}}}} \right)} \right).
        \nonumber
        \end{equation}
        When $i = j$, as ${{\bf{q}}_i}$ is a column of the DFT matrix ${\bf{Q}}$, \eqref{lemma 3 ci*cj} becomes
        \begin{equation}\label{lemma 3 ci*cj i=j}
            \begin{split}
            {\bf{d}}_j^H{\bf{d}}_j &= {{\bf{q}}_i}^H\left( {{{\bf{I}}_{{N_f}}} \otimes \left( {{{\bf{I}}_{{N_h}}} \otimes {{\bf{I}}_{{N_v}}}} \right)} \right){{\bf{q}}_i}=1.\\
            \end{split}
        \end{equation}
        Since ${{\bf{R}}_h}\left( {\theta ,\phi }  \right)$ in \eqref{horizontal steer vector rotation}, ${{\bf{R}}_v}\left( {\theta } \right)$ in \eqref{vertical steer vector rotation} are both unitary matrices  and ${{\bf{q}}_i}^H{{\bf{q}}_j}\left( t \right) =0, \forall i\ne j$, obviously ${\bf{d}}_i^H{\bf{d}}_j = 0 , \forall i\ne j$ holds.
    \end{proof}

Using Lemma \ref{lemma digonal}, the following expectation can be easily calculated as
\begin{equation}\label{lemma 2 C*E*S final}
\begin{gathered}
  {\mathbb{E}}{\left\| {{\mathbf{DE}}\left( t_p \right){{\left( {{\mathbf{S}}{{\left( t_p \right)}^\dag }} \right)}^T}} \right\|_2^2} \hfill \\
   = \sum\limits_{n = 1}^{{N_s}M} {\left( {\left( {\sum\limits_{j = 1}^{{N_s}} {\sum\limits_{m = 1}^M {{e^{ - jw_{j,m}^d\left( t_p \right)}}{s_{j,m}}{{\left( n \right)}^H}} } } \right) \cdot } \right.}  \hfill \\
  \left. {\left( {\sum\limits_{j = 1}^{{N_s}} {\sum\limits_{m = 1}^M {{e^{jw_{j,m}^d\left( t_p \right)}}{s_{j,m}}\left( n \right)} } } \right)} \right) \hfill \\
   = \sum\limits_{n = 1}^{{N_s}M} {{\mathbf{s}}{{\left( n \right)}^H}{\mathbf{s}}\left( n \right)}  = {N_s}M. \hfill \\ 
\end{gathered} 
\end{equation}
The equation \eqref{appendxi dl prediction deduce form for part two} can be easily calculated as Theorem \ref{Theorem Ns}.
\end{proof}
\section{Proof of Theorem \ref{Thereom M}}\label{Appendix Theorem M}
\begin{proof}
\emph{Proof:} When ${N_t},{N_f} \to \infty $, the following relationship is implicitly holds
  \begin{equation}\label{aysmptotic Np}
      {\widetilde {\bf{h}}^u}\left( t_p \right) = \sum\limits_{i = 1}^{{N_f}{N_t}} {\hat g_i^u\left( t_p \right){{\bf{q}}_i}}  = \sum\limits_{i \in {{\cal S}}} {\hat g_i^u\left( t_p \right){{\bf{q}}_i},\left| {{\cal S}} \right| = {{\cal{N}}_P}},  
  \end{equation}
 where ${\bf{q}}_i$ is $i$-th column of the DFT matrix $\bf{Q}$, $\hat g_i^u\left( t_p \right) = {{\bf{q}}_i}^H{{\bf{h}}^u}\left( t_p \right)$ and ${\cal{S}}$ is set of indices of which $\hat g_i^u\left( t_p \right) \ne 0$. Define the space generated by ${\bf{q}}_i\in{\cal{S}}$ as 
 \begin{equation}\label{appendix theorem M space define for angle-delay}
     {{{\cal U}}_P} = {\rm{span}}\left\{ {{\bf{q}}_i:i \in {{{\cal S}}}} \right\}.
 \end{equation}
Define $N_P$ as the number of non-identical angle-delay structure of the UL channel. The angle-delay structure of two non-intertwined paths $p,q$ are asymptotically orthogonal \cite{2020yinMobility}
 \begin{equation}\label{appendix M angle-delay asymptotic}
     \mathop {\lim }\limits_{{N_t},{N_f} \to \infty } \frac{{{{\left( {{\bf{r}}_p^u} \right)}^H}{\bf{r}}_q^u}}{{\sqrt {{N_t}{N_f}} }} = 0.
 \end{equation}
 According to \cite{2020yinMobility}, the angle-delay structure of two paths $p,q$ holds the orthogonality after projecting to an orthogonal space $\bf{S}$ which is a DFT matrix. In fact, the angle-delay structure ${\bf{r}}_p^u$ lies in the space ${{{\cal U}}_P}$ when ${N_t},{N_f} \to \infty $. Therefore, the following relationship holds 
\begin{equation}\label{appendix M angle-delay in d_j space}
\mathop {\lim }\limits_{{N_t},{N_f} \to \infty } \frac{{\left\| {{{\bf{U}}_p}^H{\bf{r}}_p^u} \right\|_2^2}}{{{N_t}{N_f}}} = 1,\mathop {\lim }\limits_{{N_t},{N_f} \to \infty } \frac{{\left\| {{{\bf{U}}_q}^H{\bf{r}}_p^u} \right\|_2^2}}{{{N_t}{N_f}}} = 0,
\end{equation}
where ${\bf{U}}_p$ is the sub-matrix formed by ${\bf{q}}_i,i\in{\cal{S}}_p$. The index set ${\cal{S}}_p$ is defined by
\begin{small}
\begin{equation}\label{appendix M index set of p define}
\mathop {\lim }\limits_{{N_t},{N_f} \to \infty } \left| {\frac{{{q_i}^H{\bf{r}}_p^u}}{{{N_t}{N_f}}}} \right| > 0,i \in {{\cal{S}}_p},\mathop {\lim }\limits_{{N_t},{N_f} \to \infty } \left| {\frac{{{q_i}^H{\bf{r}}_p^u}}{{{N_t}{N_f}}}} \right| = 0,i \notin {{\cal{S}}_p}.
\nonumber
\end{equation}
\end{small}
The sub-matrix ${\bf{U}}_q$ and set ${\cal{S}}_q$ are defined likewise. Furthermore, $N_P={\cal{N}}_P$ due to the orthogonality between ${\cal{S}}_p$ and ${\cal{S}}_q$.

The condition $L=1,N_L=2$ gives the lower bound of matrix pencil parameter configuration, which means that only one pole needs to be estimated, i.e., $M=1$. Then, the UL channel \eqref{ul angle-delay path selection} becomes
 \begin{equation}\label{appendix M ul prediction}
     {\widetilde {\bf{h}}^u}\left( t_p \right) = \sum\limits_{i \in {{\cal S}}} {{a^u}\left( i \right){z^u}{{\left( i \right)}^{t_p}}{{\bf{q}}_i}} .
 \end{equation}
 Comparing to \eqref{ul dft approximate}, \eqref{appendix M ul prediction} indicates that each angle-delay vector ${\bf{q}}_i$ corresponds to only one Doppler frequency shift $z^u\left(i\right)$. Given any ${N_d}$  the pole $z^u\left(i\right)$ is calculated through Algorithm \ref{alg1}, which means the Doppler frequency shift corresponding to ${\bf{q}}_i$ is obtained.
 
 Then we focus on the DL channel analysis. In Lemma \ref{lemma digonal}, we have proved that the DL angle-delay vector ${\bf{d}}_j$ shares the same orthogonality like ${\bf{q}}_i$, hence, the asymptotic properties \eqref{appendix M angle-delay asymptotic}-\eqref{appendix M index set of p define} hold for ${\bf{r}}_p^d$ and ${\bf{d}}_j$. Each ${\bf{d}}_j$ is calculated by \eqref{dl angle-delay path Proposition}. Denote the collection of all ${\bf{d}}_j$ as ${\bf{D}}_a \in {\mathbb{C}^{{N_t}{N_f} \times {N_t}{N_f}}}$ which is a unitary matrix. Similarly we define $\cal{D}$ like the $\cal{S}$ in \eqref{aysmptotic Np} and ${\cal{D}}_p$  like the ${\cal{S}}_p$ in \eqref{appendix M index set of p define}, respectively. The condition ${N_s}=N_P$ makes sure the following relationship holds 
 \begin{equation}\label{appendix M ns condition meaning}
     \left\| {{{\bf{D}}_a^H}{{\bf{h}}^d}\left( t_p \right)} \right\|^2_2{\rm{ = }}\left\| {\sum\limits_{j \in {\cal{D}},p \in {\cal{D}}_p} {\beta _p^d{e^{ - j2\pi {f^d}{\tau _p}}}{e^{jw_p^dt_p}}{{\bf{d}}_j}^H{\bf{r}}_p^d} } \right\|^2_2.
     \nonumber
 \end{equation}
 Using \eqref{dl doppler frequency}, the DL Doppler frequency shift satisfies
 \begin{equation}\label{appendix M doppler dl}
     {e^{j{w^d}\left( j \right)}} = {e^{jw_p^d}},j \in {\cal{D}},p \in {{\cal{D}}_p}.
 \end{equation}
The asymptotic performance of the DL channel prediction is 
\begin{equation}\label{appendix M projection index}
\begin{gathered}
  \mathop {\lim }\limits_{{N_t},{N_f} \to \infty } \varepsilon = \mathop {\lim }\limits_{{N_t},{N_f} \to \infty } \frac{{\left\| {{{\mathbf{D}}_a}\left( {{\mathbf{D}}_a^H{{\mathbf{h}}^d}\left( t_p \right) - {\mathbf{D}}_a^H{{\widetilde {\mathbf{h}}}^d}\left( t_p \right)} \right)} \right\|_2^2}}{{{N_t}{N_f}{{\left( {\sum\limits_P {\left| {\beta _p^d} \right|} } \right)}^2}}} \hfill \\ 
   = \mathop {\lim }\limits_{{N_t},{N_f} \to \infty } \frac{{\left\| {{\mathbf{D}}_a^H{{\mathbf{h}}^d}\left( t_p \right) - {\mathbf{D}}_a^H{{\widetilde {\mathbf{h}}}^d}\left( t_p \right)} \right\|_2^2}}{{{N_t}{N_f}{{\left( {\sum\limits_P {\left| {\beta _p^d} \right|} } \right)}^2}}}. \hfill \\
 \end{gathered}
 \end{equation}
 Using the property of norm, we can relax \eqref{appendix M projection index}
 \begin{small}
  \begin{equation}\label{appendix M projection index relax}
 \begin{gathered}
     \mathop {\lim }\limits_{{N_t},{N_f} \to \infty } \varepsilon \hfill \\
   \leqslant \mathop {\lim }\limits_{{N_t},{N_f} \to \infty } \frac{{\left\{ \begin{gathered}
  \left\| {\sum\limits_{j \in D,p \in {D_p}} {\beta _p^d{e^{ - j2\pi {f^d}{\tau _p}}}{e^{jw_p^dt_p}}{{\mathbf{d}}_j}^H{\mathbf{r}}_p^d} } \right\|_2^2 \hfill \\
   + \left\| {\sum\limits_{j \in D} {{a^d}\left( j \right){e^{j{w^d}\left( j \right)t_p}}} } \right\|_2^2 \hfill \\ 
\end{gathered}  \right\}}}{{{N_t}{N_f}{{\left( {\sum\limits_P {\left| {\beta _p^d} \right|} } \right)}^2}}} \hfill \\
   \leqslant \mathop {\lim }\limits_{{N_t},{N_f} \to \infty } \frac{{{{\left( {\sum\limits_{j \in D,p \in {D_p}} {\left| {\beta _p^d} \right|} } \right)}^2} + {{\left( {\sum\limits_{j \in D} {\left| {{a^d}\left( j \right)} \right|} } \right)}^2}}}{{{N_t}{N_f}{{\left( {\sum\limits_P {\left| {\beta _p^d} \right|} } \right)}^2}}}.\hfill \\
\end{gathered} 
\end{equation}
 \end{small}
Notice that $\frac{{{{\left( {\sum\limits_{j \in D,p \in {D_p}} {\left| {\beta _p^d} \right|} } \right)}^2} + {{\left( {\sum\limits_{j \in D} {\left| {{a^d}\left( j \right)} \right|} } \right)}^2}}}{{{{\left( {\sum\limits_P {\left| {\beta _p^d} \right|} } \right)}^2}}}$ is a constant which is independent with $N_tN_f$. Therefore, \eqref{appendix M projection index relax} becomes
\begin{equation}\label{appendix final}
\begin{array}{*{20}{l}}
  {\mathop {\lim }\limits_{{N_t},{N_f} \to \infty } \frac{{{{\left\| {{{\mathbf{h}}^d}\left( t_p \right) - {{\widetilde {\mathbf{h}}}^d}\left( t_p \right)} \right\|}^2}}}{{{{\left\| {{{\mathbf{h}}^d}\left( t_p \right)} \right\|}^2}}}} \\ 
  { \leqslant \mathop {\lim }\limits_{{N_t},{N_f} \to \infty } \frac{{{{\left( {\sum\limits_{j \in D,p \in {D_p}} {\left| {\beta _p^d} \right|} } \right)}^2} + {{\left( {\sum\limits_{j \in D} {\left| {{a^d}\left( j \right)} \right|} } \right)}^2}}}{{{N_t}{N_f}{{\left( {\sum\limits_P {\left| {\beta _p^d} \right|} } \right)}^2}}} = 0,} 
\end{array}
\end{equation}
 Then Theorem \ref{Thereom M} is proved.
\end{proof}
\section{Proof of Remarkl \ref{remark M}}\label{Appedix remark M}
\begin{proof}
\emph{Proof:} Obviously when ${N_s} = {N_f}{N_t}$, the value of $L$ will not affect ${\varepsilon}$. We only need to prove the case when $N_s$ satisfies ${N_s} < {N_f}{N_t}$. Clearly ${\widetilde {\bf{h}}_1^d}\left( t_p \right)$ is independent with $M$, only ${\widetilde {\bf{h}}_{2}^d}\left( t_p \right)$ need to be analyzed. We have proved that ${\left\| {{{\bf{D}}}{{\bf{E}}}\left( t_p \right){{\left( {{{\bf{S}}^\dag }} \right)}^T}} \right\|_2^2}{ = }N_sM$ in Appendix \ref{Appendix Theroem}.
The expectation part in \eqref{dl channel prediction error two part form} is calculated as
\begin{small}
\begin{equation}\label{Theorem M deduce expectation upper bound}
    \begin{split}
&\mathbb{E}\left\{ {\left\| {\frac{{{{\bf{h}}^d}\left( t_p \right) - \widetilde {\bf{h}}_1^d\left( t_p \right) - \widetilde {\bf{h}}_2^d\left( t_p \right)}}{{{{\bf{h}}^d}\left( t_p \right)}}} \right\|_2^2} \right\}\\
& \le \mathbb{E}\left\{ {\left\| {\frac{{{{\bf{h}}^d}\left( t_p \right) - \widetilde {\bf{h}}_1^d\left( t_p \right)}}{{{{\bf{h}}^d}\left( t_p \right)}}} \right\|_2^2} \right\} + \mathbb{E}\left\{ {\left\| {\frac{{\widetilde {\bf{h}}_2^d\left( t_p \right)}}{{{{\bf{h}}^d}\left( t_p \right)}}} \right\|_2^2} \right\}\\
& \le \mathbb{E}\left\{ {\left\| {\frac{{{{\bf{h}}^d}\left( t_p \right) - \widetilde {\bf{h}}_1^d\left( t_p \right)}}{{{{\bf{h}}^d}\left( t_p \right)}}} \right\|_2^2} \right\} + \mathbb{E}\left\{ {\left\| {\frac{{{\bf{DE}}\left( t_p \right){{\left( {{{\bf{S}}^\dag }} \right)}^T}{\bf{n}}{{\left( t_p \right)}^T}}}{{{{\bf{h}}^d}\left( t_p \right)}}} \right\|_2^2} \right\}\\
&\le \mathbb{E}\left\{ {\left\| {\frac{{{{\bf{h}}^d}\left( t_p \right) - \widetilde {\bf{h}}_1^d\left( t_p \right)}}{{{{\bf{h}}^d}\left( t_p \right)}}} \right\|_2^2} \right\} \\
&+ \mathbb{E}\left\{ {\left\| {{\bf{DE}}\left( t_p \right){{\left( {{{\bf{S}}^\dag }} \right)}^T}} \right\|_2^2\left\| {\frac{{{\bf{n}}{{\left( t_p \right)}^T}}}{{{{\bf{h}}^d}\left( t_p \right)}}} \right\|_2^2} \right\}\\
& \le \mathbb{E}\left\{ {\left\| {\frac{{{{\bf{h}}^d}\left( t_p \right) - \widetilde {\bf{h}}_1^d\left( t_p \right)}}{{{{\bf{h}}^d}\left( t_p \right)}}} \right\|_2^2} \right\} + \mathbb{E}\left\{ {\frac{{{N_s}M{\sigma ^2}}}{{\left\| {{{\bf{h}}^d}\left( t_p \right)} \right\|_2^2}}} \right\}.
    \end{split}        
\end{equation}    
\end{small}
Similarly, we may derive 
\begin{small}
\begin{equation}\label{Theorem M deduce expectation lower bound}
    \begin{split}
&\mathbb{E}\left\{ {\left\| {\frac{{{{\bf{h}}^d}\left( t_p \right) - \widetilde {\bf{h}}_1^d\left( t_p \right) - \widetilde {\bf{h}}_2^d\left( t_p \right)}}{{{{\bf{h}}^d}\left( t_p \right)}}} \right\|_2^2} \right\}\\
 &\ge \mathbb{E}\left\{ {\left\| {\frac{{{{\bf{h}}^d}\left( t_p \right) - \widetilde {\bf{h}}_1^d\left( t_p \right)}}{{{{\bf{h}}^d}\left( t_p \right)}}} \right\|_2^2} \right\} - \mathbb{E}\left\{ {\left\| {\frac{{{\bf{DE}}\left( t_p \right){{\left( {{{\bf{S}}^\dag }} \right)}^T}{\bf{n}}{{\left( t_p \right)}^T}}}{{{{\bf{h}}^d}\left( t_p \right)}}} \right\|_2^2} \right\}\\
 &\ge \mathbb{E}\left\{ {\left\| {\frac{{{{\bf{h}}^d}\left( t_p \right) - \widetilde {\bf{h}}_1^d\left( t_p \right)}}{{{{\bf{h}}^d}\left( t_p \right)}}} \right\|_2^2} \right\} - \mathbb{E}\left\{ {\frac{{{N_s}M{\sigma ^2}}}{{\left\| {{{\bf{h}}^d}\left( t_p \right)} \right\|_2^2}}} \right\}.
    \end{split}
\end{equation}
\end{small}
Thus Remark \ref{remark M} is proved.
\end{proof}

\ifCLASSOPTIONcaptionsoff
  \newpage
\fi
\bibliographystyle{IEEEtran}
\bibliography{IEEEabrv,reference}

\begin{thebibliography}{10}
\providecommand{\url}[1]{#1}
\csname url@samestyle\endcsname
\providecommand{\newblock}{\relax}
\providecommand{\bibinfo}[2]{#2}
\providecommand{\BIBentrySTDinterwordspacing}{\spaceskip=0pt\relax}
\providecommand{\BIBentryALTinterwordstretchfactor}{4}
\providecommand{\BIBentryALTinterwordspacing}{\spaceskip=\fontdimen2\font plus
\BIBentryALTinterwordstretchfactor\fontdimen3\font minus
  \fontdimen4\font\relax}
\providecommand{\BIBforeignlanguage}[2]{{%
\expandafter\ifx\csname l@#1\endcsname\relax
\typeout{** WARNING: IEEEtran.bst: No hyphenation pattern has been}%
\typeout{** loaded for the language `#1'. Using the pattern for}%
\typeout{** the default language instead.}%
\else
\language=\csname l@#1\endcsname
\fi
#2}}
\providecommand{\BIBdecl}{\relax}
\BIBdecl

\bibitem{2010MazattaMIMO}
T.~L. Marzetta, ``Noncooperative cellular wireless with unlimited numbers of
  base station antennas,'' \emph{{IEEE} Trans. Wireless Commun.}, vol.~9,
  no.~11, pp. 3590--3600, 2010.

\bibitem{2013MazattaSE}
H.~Q. Ngo, E.~G. Larsson, and T.~L. Marzetta, ``Energy and spectral efficiency
  of very large multiuser {MIMO} systems,'' \emph{{IEEE} Trans. Commun.},
  vol.~61, no.~4, pp. 1436--1449, 2013.

\bibitem{2011MazettaPilot}
J.~Jose, A.~Ashikhmin, T.~L. Marzetta \emph{et~al.}, ``Pilot contamination and
  precoding in multi-cell {TDD} systems,'' \emph{{IEEE} Trans. Wireless
  Commun.}, vol.~10, no.~8, pp. 2640--2651, 2011.

\bibitem{2013YinJSAC}
H.~Yin, D.~Gesbert, M.~Filippou, and Y.~Liu, ``A coordinated approach to
  channel estimation in large-scale multiple-antenna systems,'' \emph{{IEEE} J.
  Sel. Areas Commun.}, vol.~31, no.~2, pp. 264--273, 2013.

\bibitem{2014BlindRuf}
R.~R. Müller, L.~Cottatellucci, and M.~Vehkaperä, ``Blind pilot
  decontamination,'' \emph{{IEEE} J. Sel. Topics Signal Process.}, vol.~8,
  no.~5, pp. 773--786, 2014.

\bibitem{2018TDDvsFDD}
J.~Flordelis, F.~Rusek, F.~Tufvesson \emph{et~al.}, ``Massive {MIMO}
  performance—{TDD} versus {FDD}: What do measurements say?'' \emph{{IEEE}
  Trans. Wireless Commun.}, vol.~17, no.~4, pp. 2247--2261, 2018.

\bibitem{2013JSDM}
A.~{Adhikary}, J.~{Nam}, J.~{Ahn} \emph{et~al.}, ``Joint spatial division and
  multiplexing—the large-scale array regime,'' \emph{{IEEE} Trans. Inf.
  Theory}, vol.~59, no.~10, pp. 6441--6463, 2013.

\bibitem{2020JSDM}
Y.~Song, C.~Liu, Y.~Liu \emph{et~al.}, ``Joint spatial division and
  multiplexing in massive {MIMO}: A neighbor-based approach,'' \emph{{IEEE}
  Trans. Wireless Commun.}, vol.~19, no.~11, pp. 7392--7406, 2020.

\bibitem{2015spatial}
Z.~Jiang, A.~F. Molisch, G.~Caire, and Z.~Niu, ``Achievable rates of {FDD}
  massive {MIMO} systems with spatial channel correlation,'' \emph{{IEEE}
  Trans. Wireless Commun.}, vol.~14, no.~5, pp. 2868--2882, 2015.

\bibitem{2019mmimoCS}
B.~Wang, M.~Jian, F.~Gao \emph{et~al.}, ``Beam squint and channel estimation
  for wideband {mmWave} massive {MIMO-OFDM} systems,'' \emph{{IEEE} Trans.
  Signal Process.}, vol.~67, no.~23, pp. 5893--5908, 2019.

\bibitem{2017gaofeiJSAC}
H.~Lin, F.~Gao, S.~Jin \emph{et~al.}, ``A new view of multi-user hybrid massive
  {MIMO}: Non-orthogonal angle division multiple access,'' \emph{{IEEE} J. Sel.
  Areas Commun.}, vol.~35, no.~10, pp. 2268--2280, 2017.

\bibitem{2019hanEfficient}
Y.~Han, T.-H. Hsu, C.-K. Wen \emph{et~al.}, ``Efficient downlink channel
  reconstruction for {FDD} multi-antenna systems,'' \emph{{IEEE} Trans.
  Wireless Commun.}, vol.~18, no.~6, pp. 3161--3176, 2019.

\bibitem{2020DeepHan}
Y.~Han, M.~Li, S.~Jin \emph{et~al.}, ``Deep learning-based {FDD} non-stationary
  massive {MIMO} downlink channel reconstruction,'' \emph{{IEEE} J. Sel. Areas
  Commun.}, vol.~38, no.~9, pp. 1980--1993, 2020.

\bibitem{2019mobilityReport}
{Fraunhofer IIS} and {Fraunhofer HHI}, ``{RP-191951: Mobility enhancements for
  MIMO},'' in \emph{3GPP TSG RAN WG{\#}85}, January 2019, {Newport Beach
  California, USA}.

\bibitem{2020yinMobility}
H.~Yin, H.~Wang, Y.~Liu, and D.~Gesbert, ``Addressing the curse of mobility in
  massive {MIMO} with prony-based angular-delay domain channel predictions,''
  \emph{{IEEE} J. Sel. Areas Commun.}, vol.~38, no.~12, pp. 2903--2917, 2020.

\bibitem{2020Dataprediction}
X.~Xia, K.~Xu, S.~Zhao, and Y.~Wang, ``Learning the time-varying massive {MIMO}
  channels: Robust estimation and data-aided prediction,'' \emph{{IEEE} Trans.
  Veh. Technol.}, vol.~69, no.~8, pp. 8080--8096, 2020.

\bibitem{2017V2V}
R.~Wang, O.~Renaudin, C.~U. Bas \emph{et~al.}, ``High-resolution parameter
  estimation for time-varying double directional {V2V} channel,'' \emph{{IEEE}
  Trans. Veh. Technol.}, vol.~16, no.~11, pp. 7264--7275, 2017.

\bibitem{2012csidelay}
M.~A. Maddah-Ali and D.~Tse, ``Completely stale transmitter channel state
  information is still very useful,'' \emph{IEEE Transactions on Information
  Theory}, vol.~58, no.~7, pp. 4418--4431, 2012.

\bibitem{2014csidelay}
X.~Yi, S.~Yang, D.~Gesbert, and M.~Kobayashi, ``The degrees of freedom region
  of temporally correlated {MIMO} networks with delayed {CSIT},'' \emph{{IEEE}
  Trans. Inf. Theory}, vol.~60, no.~1, pp. 494--514, 2014.

\bibitem{2021dl_mobility}
C.~Wu, X.~Yi, Y.~Zhu, W.~Wang, L.~You, and X.~Gao, ``Channel prediction in
  high-mobility massive {MIMO}: From spatio-temporal autoregression to deep
  learning,'' \emph{{IEEE} J. Sel. Areas Commun.}, vol.~39, no.~7, pp.
  1915--1930, 2021.

\bibitem{2021KF_ML}
H.~Kim, S.~Kim, H.~Lee, C.~Jang, Y.~Choi, and J.~Choi, ``Massive {MIMO} channel
  prediction: {Kalman} filtering vs. machine learning,'' \emph{{IEEE} Trans.
  Commun.}, vol.~69, no.~1, pp. 518--528, 2021.

\bibitem{2021TVT}
J.~Tan and L.~Dai, ``Channel feedback in {TDD} massive {MIMO} systems with
  partial reciprocity,'' \emph{{IEEE} Trans. Veh. Technol.}, vol.~70, no.~12,
  pp. 12\,960--12\,974, 2021.

\bibitem{2021reconstruct}
H.~Lee, H.~Choi, H.~Kim, S.~Kim, C.~Jang, Y.~Choi, and J.~Choi, ``Downlink
  channel reconstruction for spatial multiplexing in massive {MIMO} systems,''
  \emph{{IEEE} Trans. Wireless Commun.}, vol.~20, no.~9, pp. 6154--6166, 2021.

\bibitem{2019jiangtaoJSTSP}
W.~Peng, W.~Li, W.~Wang, X.~Wei, and T.~Jiang, ``Downlink channel prediction
  for time-varying {FDD} massive {MIMO} systems,'' \emph{{IEEE} J. Sel. Topics
  Signal Process.}, vol.~13, no.~5, pp. 1090--1102, 2019.

\bibitem{reciprocity2002spatial}
K.~Hugl, K.~Kalliola, J.~Laurila \emph{et~al.}, ``Spatial reciprocity of uplink
  and downlink radio channels in {FDD} systems,'' in \emph{Proc. COST}, vol.
  273, no.~2.\hskip 1em plus 0.5em minus 0.4em\relax Citeseer, 2002, p. 066.

\bibitem{3gpp:36.897}
3GPP, \emph{Study on elevation beamforming / Full-Dimension ({FD}) Multiple
  Input Multiple Output ({MIMO}) for {LTE} (Release 13)}.\hskip 1em plus 0.5em
  minus 0.4em\relax Technical Report TR 36.897, available: http://www.3gpp.org,
  2015.

\bibitem{2017GaoAngle}
D.~Fan, F.~Gao, G.~Wang \emph{et~al.}, ``Angle domain signal processing-aided
  channel estimation for indoor {60-GHz} {TDD/FDD} massive {MIMO} systems,''
  \emph{{IEEE} J. Sel. Areas Commun.}, vol.~35, no.~9, pp. 1948--1961, 2017.

\bibitem{2020GlobalCom}
Z.~Zhong, L.~Fan, and S.~Ge, ``{FDD} massive {MIMO} uplink and downlink channel
  reciprocity properties: Full or partial reciprocity?'' in \emph{IEEE Glob.
  Commun. Conf., GLOBECOM - Proc.}, 2020, pp. 1--5.

\bibitem{1990MP}
Y.~{Hua} and T.~K. {Sarkar}, ``Matrix pencil method for estimating parameters
  of exponentially damped/undamped sinusoids in noise,'' \emph{{IEEE} Trans.
  Acoust., Speech, Signal Process.}, vol.~38, no.~5, pp. 814--824, 1990.

\bibitem{3gpp901}
3GPP, \emph{{Study on channel model for frequencies from 0.5 to 100 GHz
  (Release 16)}}.\hskip 1em plus 0.5em minus 0.4em\relax Technical Report TR
  38.901, available: http://www.3gpp.org, 2020.

\bibitem{3gpp211}
------, \emph{{NR; Physical channels and modulation (Release 16)}}.\hskip 1em
  plus 0.5em minus 0.4em\relax Technical Report TR 38.211, available:
  http://www.3gpp.org, 2021.

\bibitem{dft2002}
A.~M. {Sayeed}, ``Deconstructing multiantenna fading channels,'' \emph{{IEEE}
  Trans. Signal Process.}, vol.~50, no.~10, pp. 2563--2579, 2002.

\bibitem{2013BeamSpaceSayeed}
J.~Brady, N.~Behdad, and A.~M. Sayeed, ``Beamspace {MIMO} for millimeter-wave
  communications: System architecture, modeling, analysis, and measurements,''
  \emph{{IEEE} Trans. Antennas Propag.}, vol.~61, no.~7, pp. 3814--3827, 2013.

\bibitem{2019mobilityReportExperi}
{Fraunhofer IIS}, {Fraunhofer HHI}, and {Deutsche Telekom}, ``{RP-193072:
  Measurement results on Doppler spectrum for various {UE} mobility
  environments and related {CSI} enhancements},'' in \emph{3GPP TSG RAN
  WG{\#}86}, December 2019, {Sitges, Spain}.

\bibitem{1985mdl}
M.~Wax and T.~Kailath, ``Detection of signals by information theoretic
  criteria,'' \emph{{IEEE} Trans. Acoust., Speech, Signal Process.}, vol.~33,
  no.~2, pp. 387--392, 1985.

\bibitem{yin2021codebook}
H.~Yin and D.~Gesbert, ``A partial channel reciprocity-based codebook for
  wideband {FDD} massive {MIMO},'' \emph{{IEEE} Trans. Wireless Commun.}, 2022.

\bibitem{3gpp214}
3GPP, \emph{{NR; Physical layer procedures for data (Release 16)}}.\hskip 1em
  plus 0.5em minus 0.4em\relax Technical Report TR 38.214, available:
  http://www.3gpp.org, 2021.

\bibitem{2015CS}
Z.~Gao, L.~Dai, Z.~Wang, and S.~Chen, ``Spatially common sparsity based
  adaptive channel estimation and feedback for {FDD} massive {MIMO},''
  \emph{{IEEE} Trans. Signal Process.}, vol.~63, no.~23, pp. 6169--6183, 2015.

\bibitem{3gpp104}
3GPP, \emph{{NR; Physical layer procedures for data (Release 17)}}.\hskip 1em
  plus 0.5em minus 0.4em\relax Technical Report TR 38.104, available:
  http://www.3gpp.org, 2021.

\bibitem{2010TSPeigenvalue}
L.~Sun and M.~R. McKay, ``Eigen-based transceivers for the {MIMO} broadcast
  channel with semi-orthogonal user selection,'' \emph{{IEEE} Trans. Signal
  Process.}, vol.~58, no.~10, pp. 5246--5261, 2010.

\bibitem{2019MLcompare}
Y.~Zhu, X.~Dong, and T.~Lu, ``An adaptive and parameter-free recurrent neural
  structure for wireless channel prediction,'' \emph{{IEEE} Trans. Wireless
  Commun.}, vol.~67, no.~11, pp. 8086--8096, 2019.

\bibitem{2013TVAL3}
C.~Li, W.~Yin, H.~Jiang, and Y.~Zhang, ``An efficient augmented lagrangian
  method with applications to total variation minimization,''
  \emph{Computational Optimization and Applications}, vol.~56, no.~3, pp.
  507--530, 2013.

\bibitem{2009TVAL3}
C.~Li, W.~Yin, and Y.~Zhang, ``User’s guide for {TVAL3}: {TV} minimization by
  augmented lagrangian and alternating direction algorithms,'' \emph{CAAM
  report}, vol.~20, no. 46-47, p.~4, 2009.

\end{thebibliography}

%






%

\begin{IEEEbiography}[{\includegraphics[width=1in,height=1.25in,clip,keepaspectratio]{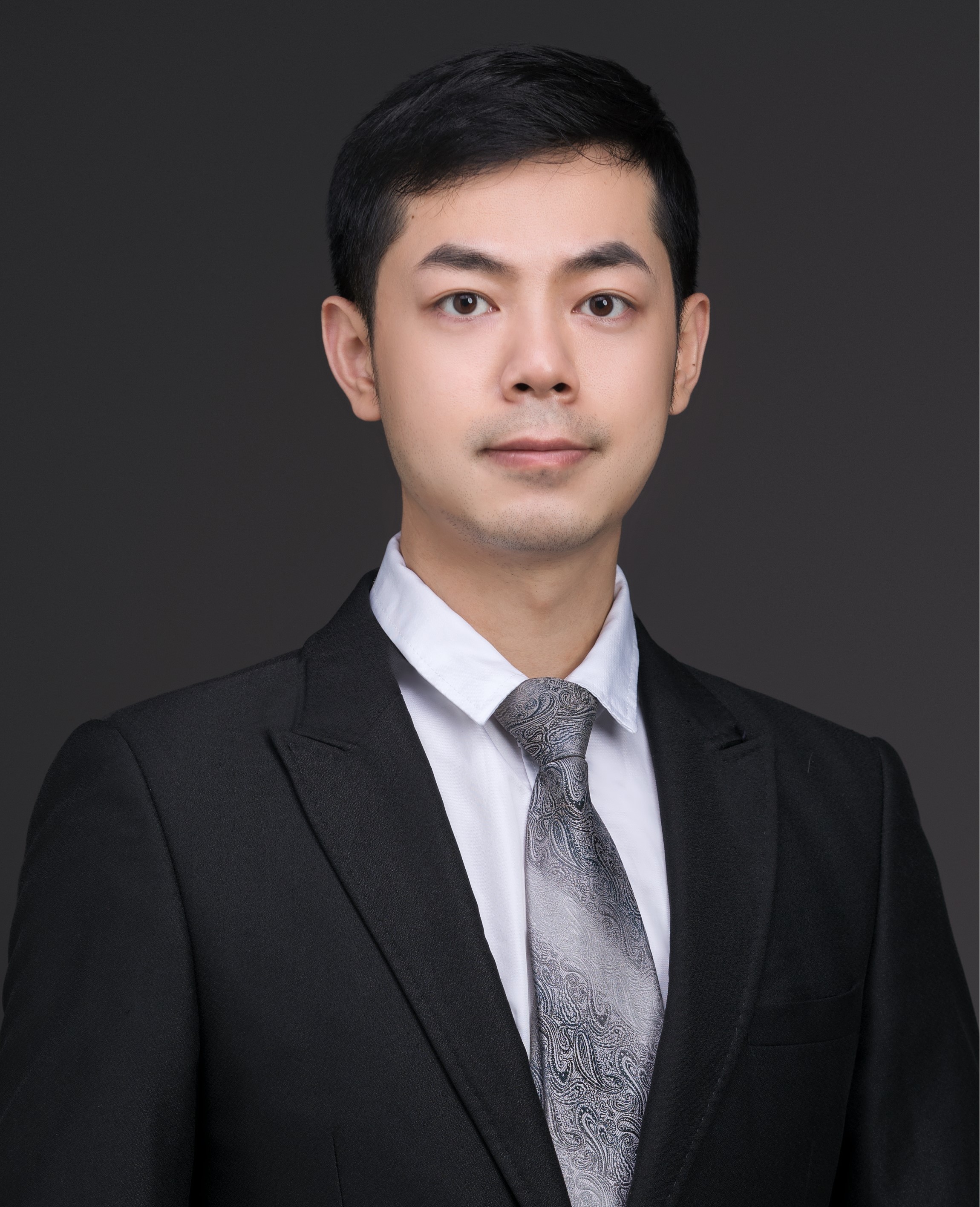}}]
{Ziao Qin} received the B.Sc. degree in Information Engineering from Beijing Institute of Technology, Beijing, China, in 2014. From 2014 to 2017, he works in industry in Beijing. Since 2018, he has been a graduate student at Huazhong University of Science and Technology, Wuhan, China. He is currently pursuing the Ph.D. degree in Information and Communications Engineering. His research interests include channel estimation, signal processing, codebook design, and beamforming for massive MIMO systems.
\end{IEEEbiography}

\begin{IEEEbiography}[{\includegraphics[width=1in,height=1.25in,clip,keepaspectratio]{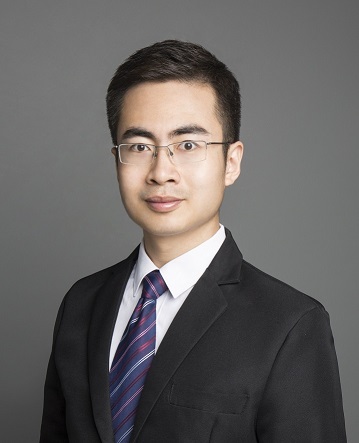}}]
{Haifan Yin} received the Ph.D. degree from T\'el\'ecom ParisTech in 2015. He received the B.Sc. degree in Electrical and Electronic Engineering and the M.Sc. degree in Electronics and Information Engineering from Huazhong University of Science and Technology, Wuhan, China, in 2009 and 2012 respectively. From 2009 to 2011, he has been with Wuhan National Laboratory for Optoelectronics, China, working on the implementation of TD-LTE systems as an R\&D engineer.
From 2016 to 2017, he has been a DSP engineer in Sequans Communications - an IoT chipmaker based in Paris, France. From 2017 to 2019, he has been a senior research engineer working on 5G standardization in Shanghai Huawei Technologies Co., Ltd., where he made substantial contributions to 5G standards, particularly the 5G codebooks. Since May 2019, he has joined the School of Electronic Information and Communications at Huazhong University of Science and Technology as a full professor. 
His current research interests include 5G and 6G networks, signal processing, machine learning, and massive MIMO systems. H. Yin was the national champion of 2021 High Potential Innovation Prize awarded by Chinese Academy of Engineering, a winner of 2020 Academic Advances of HUST, and a recipient of the 2015 Chinese Government Award for Outstanding Self-financed Students Abroad.
\end{IEEEbiography}
\begin{IEEEbiography}[{\includegraphics[width=1in,height=1.25in,clip,keepaspectratio]{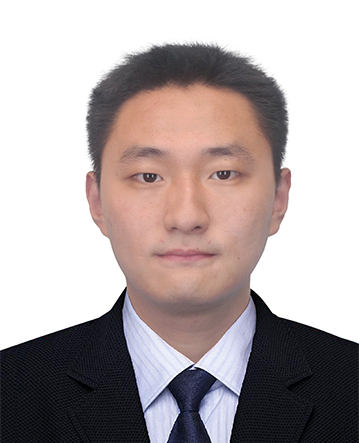}}]
{Weidong Li} received the B.Sc. degree in Electronic Information Science and Technology from Nanjing Agricultural University, Nanjing, China, in 2017, and the M.Sc. degree in Electronic Engineering from Nanjing University of Aeronautics and Astronautics, Nanjing, China, in 2020. He is currently pursuing the Ph.D. degree with the School of Electronic Information and Communications at Huazhong University of Science and Technology, Wuhan, China. His research interests include massive MIMO and signal processing. 
\end{IEEEbiography}

\begin{IEEEbiography}[{\includegraphics[width=1in,height=1.25in,clip,keepaspectratio]{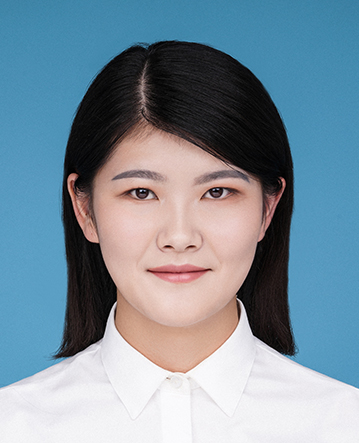}}]
{Yandi Cao} received the B.Sc. degree in Communication Engineering from Chongqing University, Chongqing, China, in 2020. She is currently pursuing the Ph.D. degree with the School of Electronic Information and Communications at Huazhong University of Science and Technology, Wuhan, China. Her research interests include massive MIMO and machine learning.  
\end{IEEEbiography}

\begin{IEEEbiography}[{\includegraphics[width=1in,height=1.25in,clip,keepaspectratio]{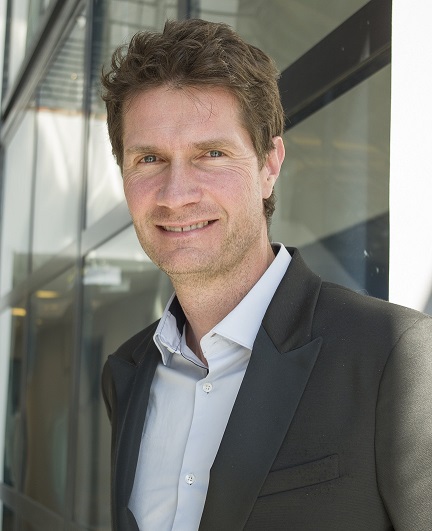}}]
{David Gesbert} (Fellow, IEEE) is Director of EURECOM, Sophia Antipolis, France (www.eurecom.fr). He received the Ph.D. degree from TelecomParis, France, in 1997. From 1997 to 1999, he was with the Information Systems Laboratory, Stanford University. He was the Founding Engineer of Iospan Wireless Inc., a Stanford spin off pioneering MIMO-OFDM (currently Intel). Before joining EURECOM in 2004, he was with the Department of Informatics, University of Oslo, as an Adjunct Professor. He has published about 350 articles and 25 patents, 7 of them winning  IEEE Best paper awards. He has been the Technical Program Co-Chair for ICC2017 and has been named a Thomson-Reuters Highly Cited Researchers in computer science.  He is a Board Member for the OpenAirInterface (OAI) Software Alliance. In 2015, he has been awarded an ERC Advanced Grant. In 2020, he was awarded funding by the French Interdisciplinary Institute on Artificial Intelligence for a Chair in the area of AI for the future IoT. In 2021, he received the Grand Prix in Research from IMT-French Academy of Sciences.
\end{IEEEbiography}
\end{document}